\documentclass[reprint,pra,groupedaddress,superscriptaddress,nofootinbib,aps,longbibliography]{revtex4-2}

\usepackage{url}
\usepackage{placeins}
\usepackage{pifont}
\usepackage{bm}
\usepackage{diagbox}
\usepackage{graphicx}
\usepackage{epsfig}
\usepackage{dsfont}
\usepackage{rotating}
\usepackage{amssymb,amsthm,amsfonts,amsbsy,amscd,amsmath,amsfonts}
\usepackage{mathrsfs}
\usepackage{bbm}
\usepackage{soul}
\usepackage{bbold}
\usepackage{times}
\usepackage[dvipsnames, svgnames, x11names]{xcolor}  
\usepackage{color}
\usepackage{physics}
\usepackage{framed}
\usepackage{enumitem}
\usepackage{epstopdf}
\usepackage{times}
\usepackage{mathtools}
\usepackage{latexsym}
\usepackage[colorlinks=true,linkcolor=teal,citecolor=purple]{hyperref}
\usepackage{cleveref}
\usepackage{lipsum}
\usepackage[ruled,vlined]{algorithm2e}
\usepackage[mathscr]{eucal}
\usepackage{eucal}
\usepackage{extarrows}
\usepackage{tikz}
\usetikzlibrary{quantikz}
\usepackage{comment}
\usepackage{array}
\usepackage{calc}
\makeatletter
\newcommand{\thickhline}{%
    \noalign {\ifnum 0=`}\fi \hrule height 1pt
    \futurelet \reserved@a \@xhline
}
\newcolumntype{"}{@{\hskip\tabcolsep\vrule width 1pt\hskip\tabcolsep}}
\makeatother

\newtheoremstyle{noparens}%
  {}{}%
  {\itshape}{}%
  {\bfseries}{.}%
  { }%
  {\thmname{#1}\thmnumber{ #2}\mdseries\thmnote{ #3}}
\theoremstyle{noparens}  
\newtheorem{lemma}{Lemma}
\newtheorem{theorem}{Theorem}
\newtheorem{corollary}{Corollary}

 \newtheorem{definition}{Definition}

\def\be{\begin{align}}
\def\ee{\end{align}}
\def\ba{\begin{array}}
\def\ea{\end{array}}
\def\Tr{\mathrm{Tr}}
\usepackage{hyperref}
\hypersetup{  colorlinks=true, linkcolor=blue, citecolor=red, urlcolor=blue  }

\usepackage{pifont}

\newcommand{\mH}{\mathcal{H}}
\newcommand{\mE}{\mathcal{E}}
\newcommand{\mS}{\mathcal{S}}

\newcommand{\mN}{\mathcal{N}}
\newcommand{\mZ}{\mathcal{Z}}
\newcommand{\mO}{\mathcal{O}}

\newcommand{\mI}{\mathbbm{I}}

\newcommand{\id}{\text{id}}

\usepackage{xcolor}

\definecolor{daxColor}{HTML}{900C3F}

\makeatletter
\newcommand*{\rom}[1]{\expandafter\@slowromancap\romannumeral #1@}
\makeatother

\usepackage{picinpar,graphicx}
\usepackage{array}

\makeatother

\usepackage{multirow}

\begin{document}

\title{Communication with Quantum Catalysts}

\author{Yuqi~Li}
\thanks{These authors contributed equally}
\address{College of Mathematics Science, Harbin Engineering University, Nantong Street, Harbin, 150001, Heilongjiang, People's Republic of China}

\author{Junjing~Xing}
\thanks{These authors contributed equally}
\address{College of Intelligent Systems Science and Engineering, Harbin Engineering University, Nantong Street, Harbin, 150001,
Heilongjiang, People's Republic of China}

\author{Dengke Qu}
\affiliation{School of Physics, Southeast University, Nanjing, 211189, People's Republic of China}
\affiliation{Beijing Computational Science Research Center, Beijing, 100193, People's Republic of China}
 
\author{Lei Xiao}
\affiliation{School of Physics, Southeast University, Nanjing, 211189, People's Republic of China}
 
\author{Zhaobing~Fan}
\email{fanzhaobing@hrbeu.edu.cn}
\address{College of Intelligent Systems Science and Engineering, Harbin Engineering University, Nantong Street, Harbin, 150001,
Heilongjiang, People's Republic of China}
\address{College of Mathematics Science, Harbin Engineering University, Nantong Street, Harbin, 150001, Heilongjiang, People's Republic of China}

\author{Zhu-Jun Zheng}
\email{zhengzj@scut.edu.cn}
\affiliation{Laboratory of Quantum Science and Engineering, South China University of Technology, Guangzhou, 510641, China}
\affiliation{Department of Mathematics, South China University of Technology, Guangzhou, 510641, China}

\author{Haitao~Ma}
\email{hmamath@hrbeu.edu.cn}
\address{College of Mathematics Science, Harbin Engineering University, Nantong Street, Harbin, 150001, Heilongjiang, People's Republic of China}

\author{Peng Xue}
\email{gnep.eux@gmail.com}
 \affiliation{School of Physics, Southeast University, Nanjing, 211189, People's Republic of China}
 \affiliation{Beijing Computational Science Research Center, Beijing, 100193, People's Republic of China}

\author{Kishor~Bharti}
\email{kishor.bharti1@gmail.com}
\address{A*STAR Quantum Innovation Centre (Q.InC), Institute of High Performance Computing (IHPC), Agency for Science, Technology and Research (A*STAR), 1 Fusionopolis Way, \#16-16 Connexis, Singapore 138632, Republic of Singapore}
\address{Centre for Quantum Engineering, Research and Education, TCG CREST, Sector V, Salt Lake, Kolkata 700091, India}

\author{Dax~Enshan~Koh}
\email{dax\_koh@ihpc.a-star.edu.sg}
\address{A*STAR Quantum Innovation Centre (Q.InC), Institute of High Performance Computing (IHPC), Agency for Science, Technology and Research (A*STAR), 1 Fusionopolis Way, \#16-16 Connexis, Singapore 138632, Republic of Singapore}

\author{Yunlong~Xiao}
\email{xiao\_yunlong@ihpc.a-star.edu.sg}
\address{A*STAR Quantum Innovation Centre (Q.InC), Institute of High Performance Computing (IHPC), Agency for Science, Technology and Research (A*STAR), 1 Fusionopolis Way, \#16-16 Connexis, Singapore 138632, Republic of Singapore}

\begin{abstract}
Communication is essential for advancing science and technology. Quantum communication, in particular, benefits from the use of catalysts. During the communication process, these catalysts enhance performance while remaining unchanged. Although chemical catalysts that undergo deactivation typically perform worse than those that remain unaffected, quantum catalysts, referred to as embezzling catalysts, can surprisingly outperform their non-deactivating counterparts despite experiencing slight alterations. In this work, we employ embezzling quantum catalysts to enhance the transmission of both quantum and classical information. Our results reveal that using embezzling catalysts augments the efficiency of information transmission across noisy quantum channels, ensuring a non-zero catalytic channel capacity. Furthermore, we introduce catalytic superdense coding, demonstrating how embezzling catalysts can enhance the transmission of classical information. Finally, we explore methods to reduce the dimensionality of catalysts, a step toward making quantum catalysis a practical reality.

\end{abstract}


\maketitle

\section{Introduction}
Communication serves as the cornerstone of our digital existence, playing an indispensable role in shaping the way we work, engage, and drive innovation in the dynamic landscape of today’s world. When integrated with the principles of quantum theory, information and communications technology undergoes a transformative evolution, unlocking unparalleled advantages in terms of security and efficiency. For example, by utilizing a quantum channel, we can facilitate the point-to-point transmission of quantum information~\cite{PhysRevA.55.1613,khatri2020principles}. Leveraging the phenomenon of entanglement, we can employ superdense coding to enhance our capacity for transmitting classical information twofold~\cite{PhysRevLett.69.2881}. These protocols pave the way for the creation of a hybrid-networked communication system, one that merges classical and quantum networks~\cite{Knaut2024,Liu2024} to revolutionize the information and communications technology with its potential to dramatically enhance data processing and transmission capabilities.

In practice, the performance of quantum communications through channels and entanglement-assisted classical communications is hindered by noise arising from imperfections in quantum devices and environmental decoherence. This results in a noisy quantum channel and imperfectly entangled states shared between the sender and receiver, thereby degrading the performance of transmitting both quantum and classical information. To address these effects, a novel approach -- termed catalytic quantum communication -- has been introduced. Drawing inspiration from the field of chemistry, this technique employs a `catalyst,' akin to biological enzymes, which consists of an additional entangled state pair~\cite{PhysRevLett.83.436,klimesh2007inequalities,PhysRevLett.113.150402, PhysRevLett.115.150402,PhysRevX.8.041051,PhysRevLett.121.190504,PhysRevLett.122.210402,PhysRevLett.123.020403,PhysRevLett.127.160402}.  These catalyst states become interlinked with the quantum resources shared by the sender and receiver, enhancing performance of the communication. Notably, the catalyst remains unaffected by the communication process, preserving its potential for subsequent uses.

Catalytic deactivation~\cite{FORZATTI1999165}, an unavoidable aspect of catalytic processes, arises from a multitude of factors, including poisoning, fouling, carbon deposition, thermal degradation, and sintering. While typically viewed as a negative occurrence, certain forms of deactivation, such as poisoning, can unexpectedly enhance the selectivity of a catalyst, as seen with the Lindlar catalyst~\cite{carey2007advanced}. This phenomenon invites us to consider the potential benefits of intentional catalytic deactivation within quantum information processes. Consider a scenario in catalytic communication where we permit the catalytic system to undergo slight alterations following its interaction with a noisy quantum channel or entangled state. This approach, termed ‘embezzling catalyst,’ poses an intriguing inquiry: could such a strategy lead to advantageous outcomes for transmitting both quantum and classical information?

In this work, we address the question by enhancing the transmission performance for both quantum and classical information through the development of embezzling catalysts. Specifically, for quantum communication, we explore the catalytic channel capacity of a noisy quantum channel, showcasing the efficacy of embezzling catalysts through our findings. For instance, we demonstrate that conventional quantum catalysts yield a catalytic channel capacity of zero for two copies of dephasing channels. In contrast, the integration of embezzling catalysts elevates this capacity to a non-zero value. In the context of long-distance entanglement distribution, we also illustrate how embezzling catalysts outperform conventional quantum catalysts. When it comes to transmitting classical information, we introduce the catalytic superdense coding, marking an improvement over traditional superdense coding without catalysts. Moreover, we propose strategies to minimize the dimensionality of embezzling catalysts, thus making quantum catalysis more practical for real-world applications. Our research paves the way for future explorations into the universal applicability and dimensional constraints of quantum catalysis in communications, potentially advancing the field.

\section{Preliminaries}\label{sec:prelims}

In this section, we will explore the foundational concepts and various classifications of catalysts. Additionally, we will present methodologies for quantifying the distinguishability of states within communication-related tasks.


\subsection{Quantum Catalysts}
Just as catalysts facilitate chemical reactions and enzymes accelerate biological processes, quantum catalysts enable certain transformations that are otherwise unattainable. Specifically, within entanglement theory, there exist pairs of bipartite states, denoted as $\rho$ and $\sigma$, between which transformation under local operations and classical communication (LOCC) is impossible; that is
\begin{align}\label{eq:ent-no-locc}
\rho\stackrel{\text{LOCC}}{\nrightarrow}\sigma
\quad
\text{and}
\quad
\rho\stackrel{\text{LOCC}}{\nleftarrow}\sigma.
\end{align}
Yet, the introduction of an ancillary state $\tau$, serving as a catalyst, renders this transformation possible.
\begin{align}\label{eq:ent-cata-locc}
    \rho\otimes\tau
    \stackrel{\text{LOCC}}{\rightarrow}
    \sigma\otimes\tau.
\end{align}
Upon completion of this catalytic process, the system under consideration, $\sigma$, and the catalyst, $\tau$, remain uncorrelated, with the catalyst retaining its original state. This type of catalyst is termed an `{\it exact catalyst}.' Should the aforementioned constraints not be satisfied, the catalyst is referred to as an `{\it approximate catalyst}.' The latter is pivotal in quantum communication tasks, including teleportation and asymptotic entanglement transformation. The formal definition of an approximate catalyst is as follows. 

\begin{definition}
[{\bf (Approximate Catalysts)}]
\label{def:AC}
A state transition from $\rho_A$ to $\sigma_A$ on system $A$ is considered approximately catalytic with respect to a set of free operations $\mO$ if, for a given catalytic system $C$ in state $\tau_C$, there exists a free operation $\Lambda\in \mO$ such that
\begin{align}\label{eq:D-state}
    D(\Lambda(\rho_A\otimes \tau_C), \sigma_A\otimes\tau_C)
    \leqslant 
    \varepsilon, 
\end{align}
and
\begin{align}\label{eq:D-catalyst}
    D(\Tr_A[\Lambda(\rho_A\otimes \tau_C)], \tau_C)
    \leqslant 
    \delta.
\end{align}
In Eqs.~\eqref{eq:D-state} and~\eqref{eq:D-catalyst}, $\varepsilon > 0$ and $\delta \geqslant 0$ serve as the smoothing parameters that define the degree of approximation permissible. The function $D$ represents a metric for assessing the distance between quantum states. The state $\tau_C$ is designated as an approximate catalyst because it facilitates the transition between states $\rho_A$ and $\sigma_A$ while undergoing only slight changes, constrained by the smoothing parameters.
\end{definition}

The term used to describe a catalyst varies based on the value of $\delta$. For small $\delta>0$, the term {\it embezzling catalyst} is applied, signifying that the system derives advantages, albeit at a marginal expense to $\tau_C$. In the case where $\delta=0$, it is termed a {\it correlated catalyst}, denoting its ability to aid in the transformation while remaining correlated with the system. For an in-depth understanding, readers are encouraged to consult the extensive literature reviews available on this subject~\cite{Datta_2023,lipkabartosik2023catalysis}. We present a detailed taxonomy that classifies catalysts in quantum information theory, as outlined in TABLE~\ref{tab:cat-classification}.

\begin{table}[b]
    \centering
    \includegraphics[width=0.48\textwidth]{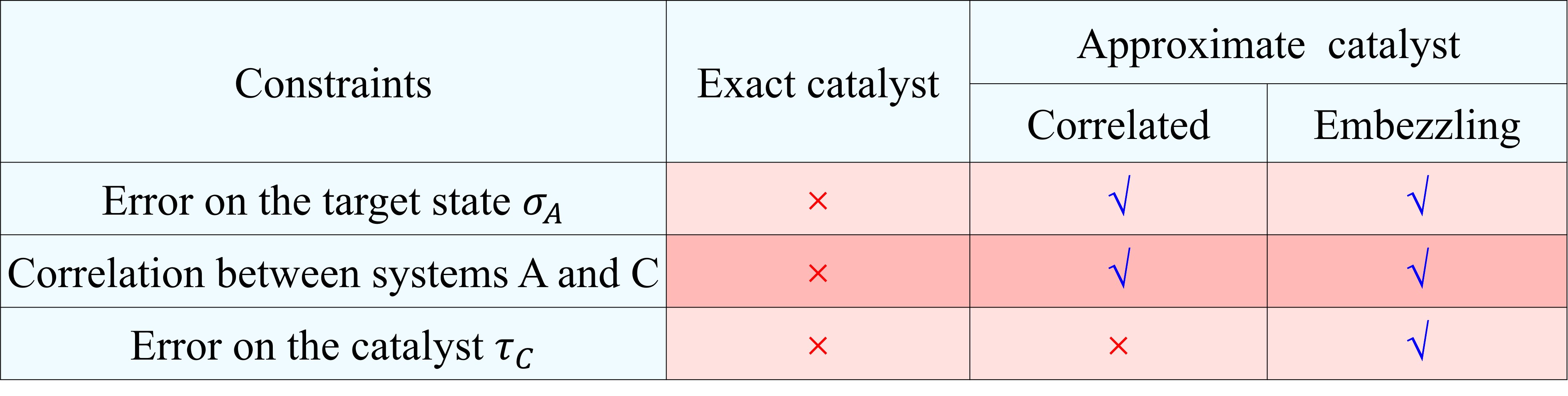}
    \caption{{ Classification of catalysts.} 
    An exact catalyst is a quantum state that enables a transformation that would otherwise be impossible, without itself changing or becoming entangled with other systems. Consider two quantum states, $\rho_A$ and $\sigma_A$, and a set of permissible operations $\mO$. If a direct transformation from $\rho_A$ to $\sigma_A$ is not feasible under $\mO$, namely $\Lambda(\rho_A)\neq\sigma_A$ for any free operation $\Lambda\in\mO$ (see Eq.~\eqref{eq:ent-no-locc}), the role of an exact catalyst becomes evident. Should there exist a state $\tau_C$ that allows the combined system $\rho_A \otimes \tau_C$ to transform into $\sigma_A \otimes \tau_C$ via $\mathcal{O}$ (see Eq.~\eqref{eq:ent-cata-locc}), then $\tau_C$ serves as an exact catalyst for the transition from $\rho_A$ to $\sigma_A$~\cite{PhysRevLett.83.3566, PhysRevA.64.042314, PhysRevA.71.042319,Turgut_2007, Marvian_2013, doi:10.1038/ncomms7383, PhysRevA.93.042326}. We can broaden the criteria for exact catalysts to establish a more encompassing concept of approximate catalysts, as delineated in Def.~\ref{def:AC}. These approximate catalysts fall into two categories: correlated and embezzling. Correlated catalysts operate without consuming the catalyst at all~\cite{PhysRevLett.126.150502,PhysRevLett.127.150503,PhysRevLett.127.260402,Yadin_2022,PhysRevLett.129.120506,Wilming2022correlationsin,PhysRevA.107.012404,PhysRevLett.130.040401,PhysRevLett.130.240204,ganardi2023catalytic,Datta2024entanglement}, while embezzling catalysts~\cite{PhysRevA.67.060302,PhysRevLett.111.250404,doi:10.1073/pnas.1411728112,doi:10.1073/pnas.1411728112,doi:10.1007/s00037-015-0098-3,Ng_2015,7377103,10.1063/1.4938052,PhysRevLett.118.080503,PhysRevA.100.042323,PhysRevX.11.011061,doi:10.1038/s41534-022-00608-1,Luijk2023covariantcatalysis,PRXQuantum.4.040330,Zanoni2024complete,vanluijk2024embezzling,vanluijk2024embezzlement} are notable for their slight consumption of the catalyst during use.
    }
    \label{tab:cat-classification}
\end{table}


\subsection{Distinguishability Measures}
Utilizing a suitable distinguishability measure is crucial for understanding the evolution of states under free transformation and the dynamics of catalysts in quantum information processing. Here, we will explore a variety of distinguishability measures essential to our analysis. The foremost among these is the {\it Uhlmann’s fidelity}~\cite{doi:10.1080/09500349414552171}. For states $\rho$ and $\sigma$ acting on the same system, the Uhlmann’s fidelity $F_U$ between them is given by
\begin{align}\label{eq:uhl}
	 F_U(\rho,\sigma) 
      :=
      \left\|\sqrt{\rho}\sqrt{\sigma}\right\|_{1}^{2}
      =
      \left(\text{Tr}\Big[
      \sqrt{\sqrt{\sigma}\rho\sqrt{\sigma}}
      \Big]\right)^{2}.
\end{align}
For the special case of a pure state, where $\sigma= \ketbra{\psi}$, the aforementioned Eq.~\eqref{eq:uhl} simplifies to
\begin{align}
F_{U}(\rho, \ketbra{\psi}{\psi})=\Tr[\rho\cdot \ketbra{\psi}{\psi}] = \bra\psi \rho\ket\psi.
\end{align}
While the Uhlmann fidelity does not fulfill the criteria of a mathematical distance due to its value of one for identical states, it gives rise to a metric known as the {\it purified distance}~\cite{tomamichel2013framework}, which is also referred to as the {\it sine distance}~\cite{rastegin2006sine}. Formally, the purified distance is defined by
\begin{align}\label{eq:pur}
P(\rho, \sigma)
:= 
\sqrt{1-F_{U} (\rho , \sigma )},
\end{align}
which meets all the conditions required of a mathematical distance, including the triangle inequality. Specifically, for any quantum states $\rho$, $\sigma$, and $\tau$, we have
\begin{align}\label{eq:pur-tri}
    P(\rho, \sigma)
    \leqslant  
    P(\rho, \tau)+P(\tau, \sigma).
\end{align}

An alternative measure for the distinguishability between states is the quantum max-relative entropy $D_{\max}$, which offers a valuable substitute for the conventional quantum relative entropy in single-shot scenarios of quantum information processing tasks. For states $\rho$ and $\sigma$, given that the support of $\rho$ is included in the support of $\sigma$, namely ${\rm supp} (\rho) \subseteq {\rm supp} (\sigma)$, the max-relative entropy of $\rho$ with respect to $\sigma$ is defined as follows
\begin{align}
D_{\max}(\rho\,\|\,\sigma)
:=
\inf\{\lambda: \rho\leqslant2^{\lambda}\sigma\}.
\end{align}
It is also noteworthy that the max-relative entropy can be calculated via a semidefinite programming (SDP). In particular, the max-relative entropy corresponds to the logarithm of the optimal value obtained from the SDP
\begin{align}\label{eq:D-max}
    D_{\max}(\rho\,\|\,\sigma)
    =
    \log_2 \max_{M\geqslant 0}\{\Tr[M\rho]: \Tr[M\sigma]\leqslant 1\}.
\end{align}
Considering that SDPs can be effectively solved using {\it interior point methods}, as detailed in Refs.~\cite{KHACHIYAN198053,doi:10.1137/1038003,Boyd_Vandenberghe_2004}, a vast array of SDP applications~\cite{xiao2019complementary,PhysRevResearch.3.023077,10.1088/978-0-7503-3343-6,PhysRevLett.130.240201,Yuan2023}, including the computation of max-relative entropy, can typically be executed with considerable efficiency in practical scenarios.

\section{Catalytic Quantum Information Transmission}\label{sec:cqc}

\begin{table}[t]
    \centering
    \includegraphics[width=0.48\textwidth]{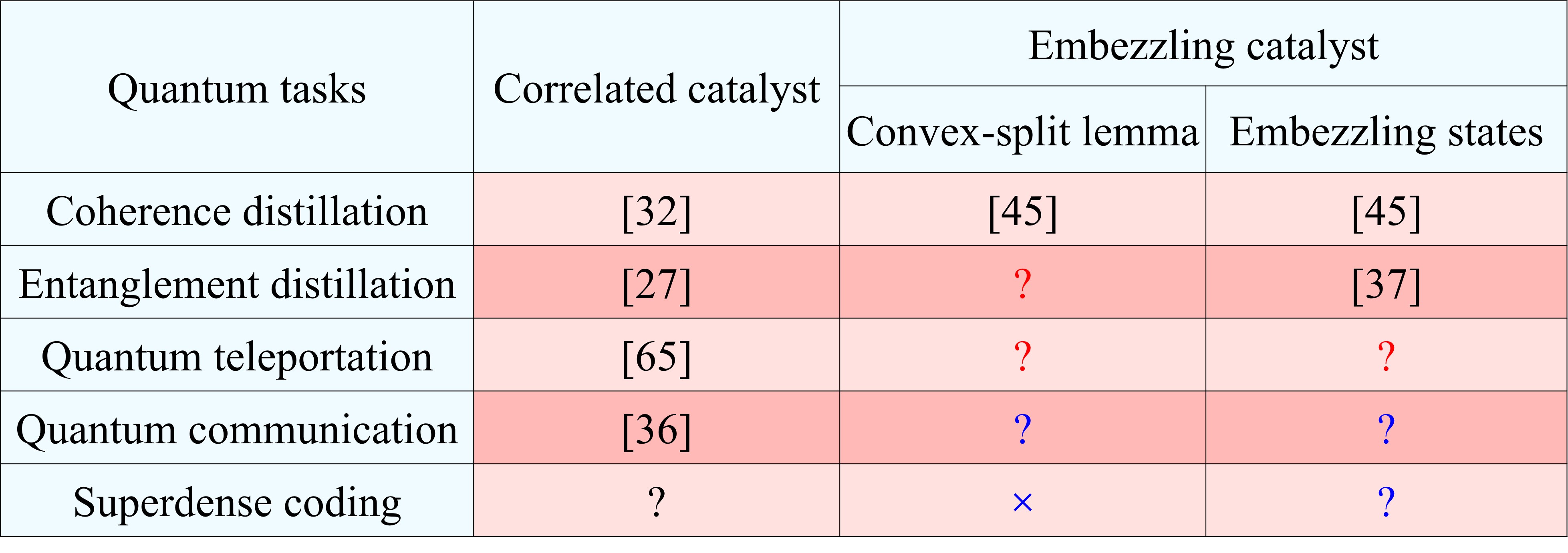}
    \caption{{Application scope of catalysts.} The blue question marks denote the tasks that will be explored and addressed in this work, while the red question marks represent tasks planned for investigation and resolution in our accompanying paper~\cite{CP}. The cross signifies that the convex-split-lemma-assisted protocol does not work for superdense coding, and the black question mark indicates that it remains unknown whether correlated catalysts exist for superdense coding.}
    \label{tab:cat-application}
\end{table}

Quantum communication harnesses the principles of quantum physics to transmit information encoded in qudits. The capacity of a noisy quantum channel to convey information is a key concept in quantum information theory, delineating the maximum feasible communication rates. Here, we focus on scenarios involving additional catalysts (see TABLE~\ref{tab:cat-application}). More specifically, we investigate the ability to generate entanglement through a single use of a noisy quantum channel, represented by $\mN$, with the assistance of catalysts distributed between the sender, Alice, and the receiver, Bob.

The schematic diagram of the protocol for catalytic quantum communication is illustrated in Fig.~\ref{fig:cqc}. This approach diverges from conventional methods by allowing classical communication between Alice and Bob, in addition to the use of a noisy quantum channel $\mN$. Specifically, Alice starts with two particles in the state $\ket{00}_{AA'}$ and shares a catalyst $\tau_{CC'}$, with Bob. Before utilizing the noisy quantum channel $\mN$, they are free to perform any LOCC operations as pre-processing. Following the use of the channel $\mN$, Alice and Bob can engage in LOCC post-processing to derive an output state $\mu_{ABCC'}$, which closely resembles the tensor product of a maximally entangled state $\phi^+$ acting on systems $AB$ and the catalyst $\tau$ acting on systems $CC'$. The catalytic channel capacity is defined by the maximum number of ebits that can be generated through this process.

\begin{figure}[t] 
	\begin{center}
		\includegraphics[width=0.48\textwidth]{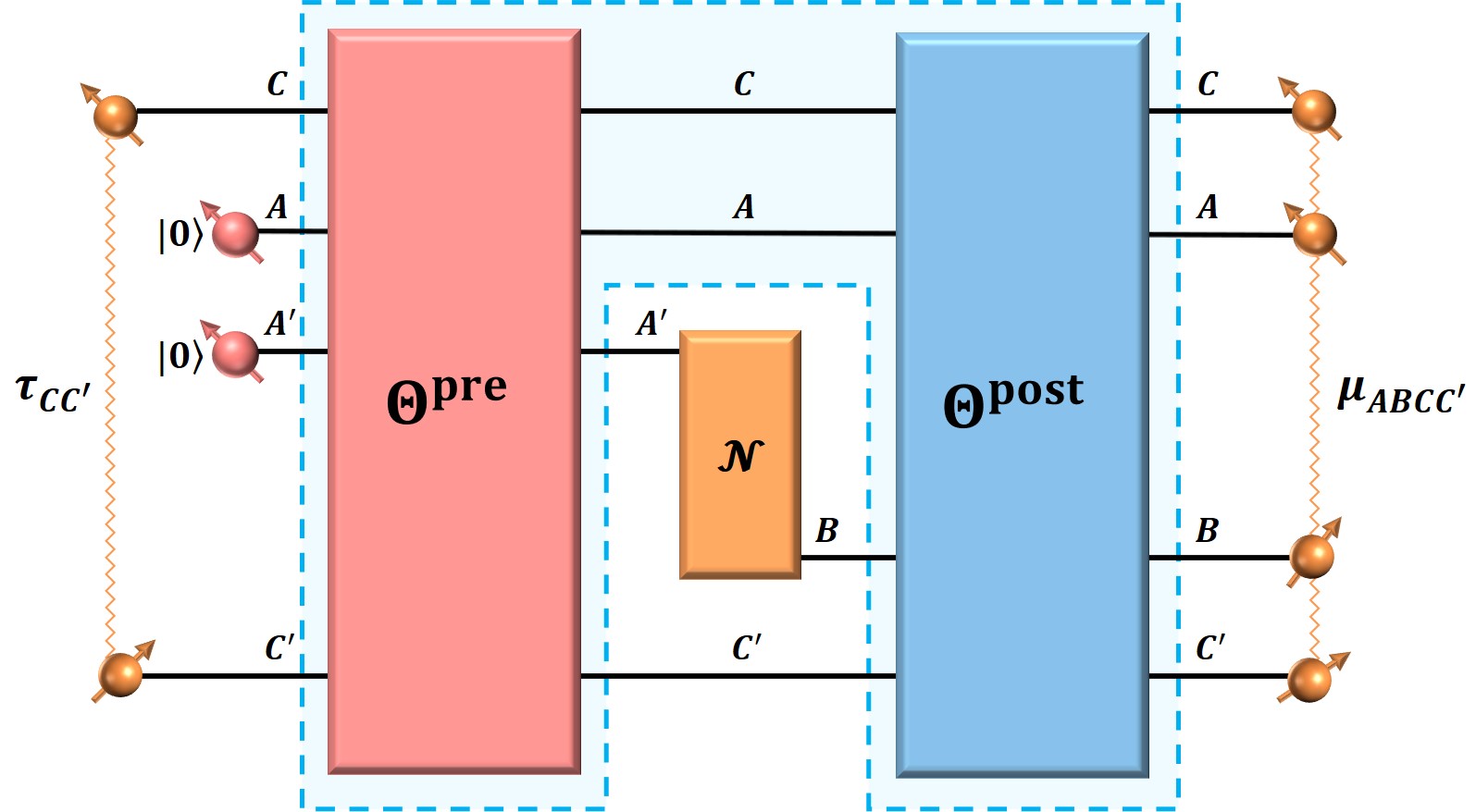}
	\end{center} 
	\caption{{Catalytic quantum communication.} 
    A noisy quantum channel $\mN$ establishes a connection between the sender, Alice, and the receiver, Bob. To facilitate the generation of shared entanglement via the channel, Alice and Bob utilize a bipartite state $\tau_{CC'}$ as a catalyst. They implement LOCC operations as pre-processing $\Theta^{\text{Pre}}$ and post-processing $\Theta^{\text{Post}}$ steps, respectively, before and after the operation of the noisy channel $\mN$. To maintain consistency with the original work on catalysts for noisy quantum channels~\cite{Datta2024entanglement}, we assume that Alice will prepare states initialized in $\ket{00}_{AA'}$. However, it is important to note that since implementing LOCC is cost-free in this protocol, the preparation can be integrated into the pre-processing operation as $\Theta^{\text{Pre}}:= \mE\circ\ket{00}_{AA'}\in\text{LOCC}$ holds for any $\mE\in\text{LOCC}(AC:A'C')$. The state $\mu_{ABCC'}$ produced by catalytic quantum communication will approximate the tensor product of a maximally entangled state $\phi^+_d$ and the catalyst $\tau_{CC'}$, that is, $\mu_{ABCC'}\approx\phi^+_{d, AB}\otimes\tau_{CC'}$. In this context, the composition of LOCC pre-processing $\Theta^{\text{Pre}}$ and post-processing $\Theta^{\text{Post}}$ forms an LOCC superchannel.}
     \label{fig:cqc}
\end{figure}

\begin{definition}
[{\bf (Catalytic Channel Capacity)}]
\label{def:CCC}
Given a noisy quantum channel $\mN$ and a non-negative error threshold $\epsilon$, the single-shot $\epsilon$-error catalytic channel capacity of $\mN$, denoted by $Q_c^{\epsilon}(\mN)$, is the supremum of $\log d$ -- the number of qubits that can be reliably transmitted -- across all catalytic LOCC superchannels $\Theta^{\rm{Post}}\circ\Theta^{\rm{Pre}}\circ\tau$, as depicted in Fig.~\ref{fig:cqc}; that is
\begin{align} 
\label{eq:capacity}
Q_c^{\epsilon}(\mN)
:=
\sup
\,
&\log d
\notag\\
\rm{s.t.}
\,\,
&\,\,
p_{\rm{err}}^c(\mu)
:=
1-
F_U(\mu, \phi^+_{d}\otimes\tau)\leqslant\epsilon,\notag\\
&\,\,
\mu:=
\Theta^{\rm{Post}}
\circ
(\rm{id}\otimes\mN)
\circ
\Theta^{\rm{Pre}}
(\ketbra{00}\otimes\tau),\notag\\
&\,\,
\Theta^{\rm{Pre}},\,\Theta^{\rm{Post}}\in\rm{LOCC}.
\end{align}
Since the systems associated with each state and quantum channel are clearly indicated in our Fig.~\ref{fig:cqc}, we have omitted them here to avoid redundancy and maintain clarity. If the catalyst system is not utilized, the state resulting from the noisy quantum channel under LOCC superchannel is denoted as $\sigma_{AB}$, or simply $\sigma$. In this setting, we evaluate the efficacy of $\mN$ in generating entanglement by the following measure
\begin{align}
\label{eq:err}
p_{\rm{err}}(\sigma)= 1-F_U(\sigma, \phi^+_{d}).
\end{align}
\end{definition}

Above definition broadens the scope of catalytic channel capacity introduced in Ref.~\cite{Datta2024entanglement}. It incorporates the potential for errors in quantum information processing, thereby allowing for a wider variety of catalysts, such as embezzling states. When the dimension of the catalyst tends towards infinity and error $\epsilon$ approaches zero, $Q_c^{\epsilon}$ converges to the definition explored in Ref.~\cite{Datta2024entanglement}. 

In scenarios involving a correlated catalyst, where the marginal of the joint system remains unchanged under catalytic quantum information processing, as shown in Fig.~\ref{fig:cqc}, certain noisy quantum channels may exhibit zero catalytic channel capacity, even when an infinite-dimensional catalyst is used. Take, for instance, the dephasing channel 
\begin{align}\label{eq:dephasing}
    \mZ_p(\cdot)
    := 
    p\,\id(\cdot)+ (1-p)\,\sigma_{z}(\cdot)\sigma_{z}.
\end{align}
Specifically, when $p<0.817$, the catalytic channel capacity $Q_c^{0}(\mZ_p^{\otimes2})$ is zero (see Ref.~\cite{Datta2024entanglement}). However, by employing an embezzling catalyst, a non-zero catalytic channel capacity can be achieved for $\mZ_p^{\otimes2}$ even when $p<0.817$, highlighting the advantage of utilizing embezzling catalysts in quantum communications. The formal results for the general catalytic channel capacity are as follows.

\begin{theorem}
[{\bf (Non-Zero Catalytic Channel Capacity)}]
\label{thm:cqc}
For any noisy quantum channel $\mN$ and a given positive error threshold $\epsilon >0$, it is always possible to identify quantum catalysts with finite dimensions that ensure the single-shot $\epsilon$-error catalytic channel capacity is non-zero, namely
\begin{align}
         Q_c^{\epsilon}(\mN)>0.
\end{align}
\end{theorem} 

Directly computing the catalytic channel capacity $Q_c^{\epsilon}(\mN)$ of a noisy quantum channel $\mN$ poses significant challenges due to the intricate nature of LOCC, whose mathematical structure resists straightforward characterization. This complexity extends to optimizations involving LOCC. The introduction of quantum catalysts adds another layer of difficulty. To overcome these obstacles, we consider a specific pre-processing operation $\Theta^{\rm{Pre}}$ that prepares a maximally entangled state $\phi^+_{AA^{'}}$. We then construct the Choi-Jamio\l kowski state $\rho_{AB}^{\mN}$ of channel $\mN$ by sending one half of the maximally entangled state through the noisy channel $\mN$, as follows
\begin{align}\label{eq:choi}
    \rho_{AB}^{\mN}:=
    (\rm{id}_A\otimes\mN_{A^{'} \to B})
    \phi^+_{AA^{'}}.
\end{align} 
We introduce several specific protocols that provide computable bounds, effectively establishing lower bounds for the catalytic channel capacity. In numerical experiments, we use various embezzling catalysts, denoted as $\tau$, and consider the corresponding LOCC operation $\mE_{\tau}: ABCC^{'}\to ABCC^{'}$ in post-processing. Under this operation, we can derive a state for the system of interest, i.e.,
\begin{align}\label{eq:cata}
    \rho_{AB}^{\mN, \text{Cat}}:=
    \Tr_{CC^{'}}\left[
    \mE_{\tau}
    \left(
    \rho_{AB}^{\mN}
    \otimes
    \tau_{CC^{'}}
    \right)\right].
\end{align}
with higher entanglement fidelity, or equivalently, reduced error. Here, the map $\mE_{\tau}$ varies depending on the choice of $\tau$. We will investigate the decrease in error achieved by using catalysts in communication and use this to benchmark the performance of catalytic protocols. Mathematically, the reduction in error is characterized by
\begin{align}\label{eq:delta-p-err}
    \Delta p_{\text{err}}
    :=\,\,
    &p_{\rm{err}}(\rho_{AB}^{\mN})- 
    p_{\rm{err}}(\rho_{AB}^{\mN, \text{Cat}})\notag\\
    =\,\,
    &
    F_U(\rho_{AB}^{\mN, \text{Cat}}, \phi^+_{d})-
    F_U(\rho_{AB}^{\mN}, \phi^+_{d})\notag\\
    \geqslant\,\,
    &
    F_U(\mE_{\tau}
    \left(
    \rho_{AB}^{\mN}
    \otimes
    \tau_{CC^{'}}
    \right), \phi^+_{d}\otimes\tau)-
    F_U(\rho_{AB}^{\mN}, \phi^+_{d}).
\end{align}
More insights and detailed explanations will be elaborated in the subsequent sections.

\subsection{Convex-Split-Lemma-Assisted Quantum Communication}
Our first catalytic quantum communication protocol integrates the {\it convex-split lemma} with the framework of quantum superchannels. Originating in Ref.~\cite{PhysRevLett.119.120506} to analyze communication costs, the convex-split lemma has evolved into an indispensable tool for a wide range of tasks in quantum information processing~\cite{8123869,8234697,8399830}. Its applications encompass various areas, including catalytic decoupling~\cite{PhysRevLett.118.080503}, quantum resource theories~\cite{PhysRevX.11.011061,PhysRevA.100.042323}, and single-shot quantum communication~\cite{doi:10.1038/s41534-022-00608-1}. We now proceed to examine its formulation in detail.

\begin{lemma}
[{\bf(Convex-split lemma~\cite{PhysRevLett.119.120506})}]
\label{lem:CS}
Let $\rho$ and $\tau$ represent quantum states such that the support of $\rho$ is contained within the support of $\tau$, i.e., ${\rm supp} (\rho) \subseteq {\rm supp} (\tau)$. Define $k$ as the max-relative entropy of $\rho$ with respect to $\tau$, namely $k:= D_{\rm{max}}(\rho||\tau)$. Now, consider the following quantum state
\begin{align}\label{eq:U}
	\mu
        = 
        \frac{1}{n}\sum_{i=1}^n
        \tau_{1}\otimes\cdots\otimes
        \rho_{i}\otimes
        \cdots\otimes\tau_{n},
\end{align}
where $\tau_i=\tau$ and $\rho_i=\rho$ for any $i$ within the set $[n]:=\{1, \ldots, n\}$. Then, the purified distance (see Eq.~\eqref{eq:pur}) between $u$ and $n$ copies of $\tau$ is upper bounded by
\begin{align}\label{eq:rho-tau}
         P(\mu,\tau^{\otimes n})
          \leqslant
         \sqrt{\frac{2^k}{n}}.
\end{align}
\end{lemma}

The above lemma suggests that when the max-relative entropy of 
$\rho$ with respect to $\tau$ is fixed, employing more copies of 
$\tau$ in constructing $\mu$ will result in $\mu$ and $\tau^{\otimes n}$ becoming less distinguishable. A diagram illustrating the convex split lemma is provided in Fig.~\ref{fig:CS-channel}.
Next, we will utilize the convex-split lemma to introduce an enhanced version of Thm.~\ref{thm:cqc}, which subsequently establishes Thm.~\ref{thm:cqc} as an immediate corollary. Our proof reveals the intricate interplay between the catalytic channel capacity and the dimensions of the embezzling catalyst.


Before presenting our main theorem, let us first address the issue of the quantum system's dimension. As mentioned previously in Eq.~\eqref{eq:choi}, our protocol begins by preparing a maximally entangled state $\phi^+_{m}$ with Schmidt rank $m$. However, this $m$ does not necessarily equal the input dimension of the noisy quantum channel, namely $m\neq\dim A$. If $m>\dim A$, we can project the state onto system $A$. Conversely, if $m<\dim A$, we can use an isometry to map $\phi^+_{m}$ into system $A$. Hence, without loss of generality, we can assume $m=\dim A$, allowing the noisy channel $\mN$ to be applied directly to $\phi^+_{m}$. Now, we move on to the convex-split-lemma-assisted catalytic channel capacity theorem.

\begin{theorem}
[{\bf (Convex-Split-Lemma-Assisted Catalytic Channel Capacity)}]
\label{thm:csl}
Given a noisy quantum channel $\mN$ (see Fig.~\ref{fig:cqc}) and an error threshold $\epsilon>0$, we can construct a bipartite state $\tau:= p\phi^+_m + (1-p)\zeta$ such that $P(\tau, \phi^+_{m})\leqslant\sqrt{\epsilon}/2$. Here, $\zeta$ is a full-rank quantum state. Let $k$ denote the max-relative entropy of $\rho^{\mN}$ (see Eq.~\eqref{eq:choi}) with respect to $\tau$, that is, $k:= D_{\rm{max}}(\rho^{\mN}\parallel\tau)$. Based on $n-1$ copies of $\tau$, we can introduce the following catalyst
\begin{align}\label{eq:tau-cs}
  \tau^{CS}:= \tau^{\otimes n-1}.
\end{align}
For max-relative entropy $k$, if 
\begin{align}\label{eq:restriction-k-n}
     n \geqslant \left\lceil \frac{2^{k+2}}{\epsilon}\right\rceil,
\end{align}
then there exists a catalytic quantum communication protocol satisfying
\begin{align}
    Q_c^{\epsilon}(\mN) \geqslant  \log m.
\end{align}

\end{theorem}


\begin{figure}[t]
	\begin{center}
\includegraphics[width=0.48\textwidth]{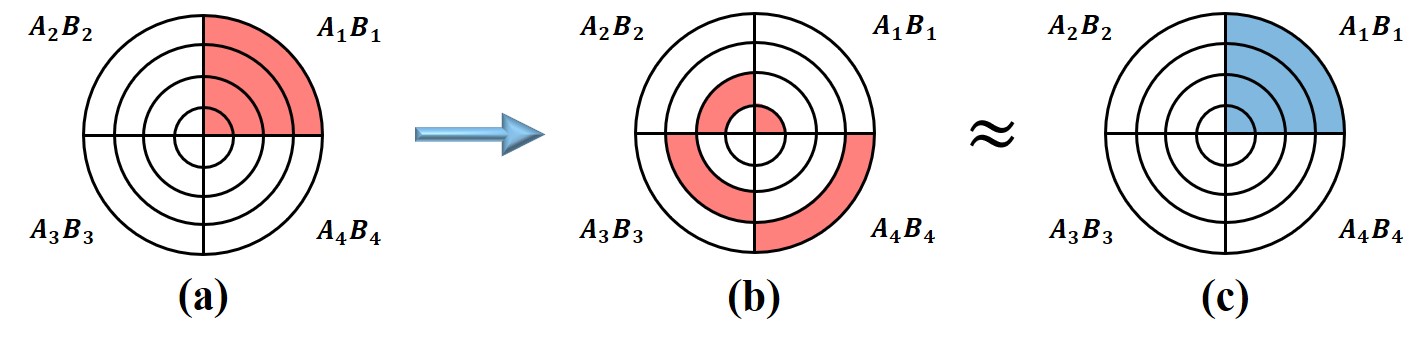}
	\end{center}
	\caption{{Graphical illustration of the convex-split lemma.} 
We use a visual example with $n=4$ to illustrate the convex-split lemma. In this example, each color represents a distinct quantum state: red signifies $\rho^{\mN}$, white denotes $\tau$, and blue indicates the maximally entangled state $\phi^+$. The diagram features a uniform mixture of four concentric circles. Each concentric circle, labeled as the $i$-th layer (where $i$ ranges from $1$ to $4$), represents a different quantum state, with the innermost circle being the 1st layer and the outermost circle being the 4th layer.
For instance, in (a), circles stand for the state of form $\rho_{1}^{\mN}\otimes \tau_{2} \otimes \tau_{3} \otimes \tau_{4}$, where $i\in[4]$ denotes systems $A_iB_i$. Meanwhile, in (b), the $i$-th layer corresponds to the quantum state $\tau_{1}\otimes\cdots\otimes\rho_{i}^{\mN}\otimes\cdots\otimes\tau_{4}$, resulting in the overall state $(1/4)\sum_{i=1}^{4}\tau_{1}\otimes\cdots\otimes\rho_{i}^{\mN}\otimes\cdots\otimes\tau_{4}$ (see Eq.~\eqref{eq:U}). Finally, in (c), circles showcase the state $\phi^+_{1}\otimes \tau_{2} \otimes \tau_{3} \otimes \tau_{4}$, which we aim to approximate using the convex-split lemma.} \label{fig:CS-channel}
\end{figure}

\begin{proof}
We proceed with the proof by constructing the catalytic quantum communication protocol. Initially, Alice creates a maximally entangled state $\phi^+_{AA'}$ locally on her side and sends a part of her system through the noisy channel $\mN$, resulting in the Choi state of $\rho^{\mN}:=\rho_{AB}^{\mN}$ (see Eq.~\eqref{eq:choi}). Incorporating the catalyst $\tau^{CS}$ (see Eq.~\eqref{eq:tau-cs}), the overall state becomes
\begin{align}\label{eq:rho^N-tau}
    \rho_1^{\mN}\otimes\tau_2\otimes\cdots\otimes\tau_n,
\end{align}
where system $1$ stands for $A_1B_1:= AB$, and $i$ represents $A_iB_i$ systems with $A_i\cong A$, $B_i\cong B$, and $2\leqslant i\leqslant n$. In this case, the catalytic systems are $CC^{'}:= \otimes_{i=2}^{n}A_iB_i$. During the post-processing of catalytic quantum communication, Alice can randomly select a number $i$ from $1$ to $n$ and send it to Bob. Upon receiving the classical message $i$, Alice and Bob will SWAP $\rho^{\mN}$ with the $i$-th system of Eq.~\eqref{eq:rho^N-tau}, resulting in the form of $\mu$ constructed in Eq.~\eqref{eq:U}. Note that within this protocol, the pre-processing is achieved through local operations, and the post-processing is a one-way LOCC. Hence, the overall procedure forms an LOCC superchannel. 

Thanks to the convex-split lemma (see Lem.~\ref{lem:CS}) and the triangle inequality of purified distance, we obtain
\begin{align}
    P(\mu, \phi^+_m\otimes\tau^{CS})
    &\leqslant
    P(\mu, \tau^{\otimes n})
    +
    P(\tau^{\otimes n}, \phi^+_m\otimes\tau^{CS})\notag\\
    &\leqslant
    \sqrt{\frac{2^k}{n}}+\frac{\sqrt{\epsilon}}{2}.
\end{align}
With sufficient copies of $\tau$ in formulating the catalyst $\tau^{CS}$ (see Eq.~\eqref{eq:tau-cs}), we arrive at
\begin{align}\label{eq:p-errc}   
         p_{\text{err}}^c(\mu)
         =
         1-F_U(\mu, \phi^+\otimes\tau^{CS})
         \leqslant
         \epsilon,
\end{align}
which concludes the proof.
\end{proof}

\begin{figure*}
    \centering
    \includegraphics[width=1\textwidth]{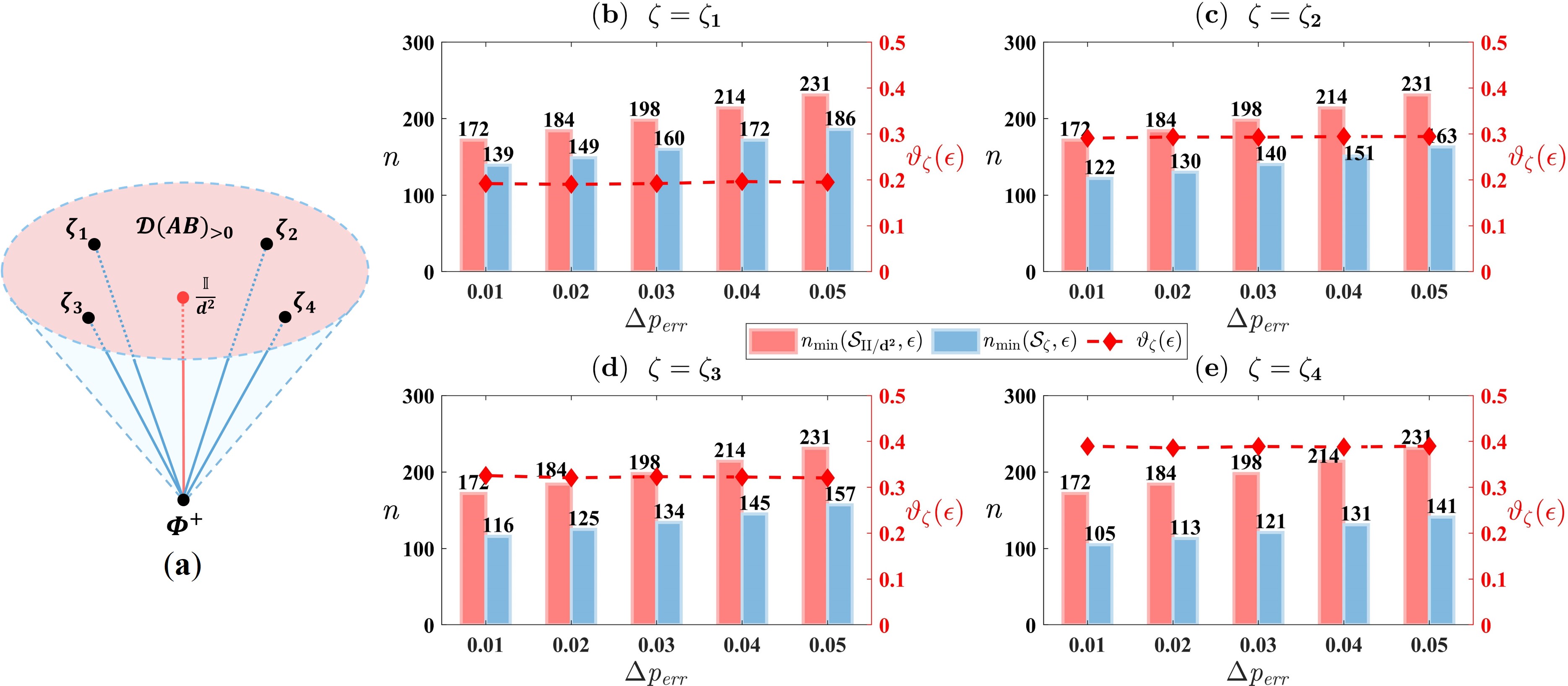}
     \caption{{Comparisons between $n_{\text{min}}(\mS_{\mI/d^2}, \epsilon)$ and $n_{\text{min}}(\mS_{\zeta_i}, \epsilon)$.} 
    In (a), we visualize our idea and depict the relationship between the maximally mixed state $\mI/d^2$ and the four randomly selected fully ranked states $\zeta_i$ for $i \in [4]$. Each point on the line represents a potential $\tau$ considered in Eq.~\eqref{eq:tau-cs}. In (b) to (e), the red and blue bar graphs represent the number of copies of $\tau$ required, with the red dotted line indicating the percentage of copy reduction, i.e., the descent ratio $\vartheta(\epsilon)$, as defined in Eq.~\eqref{eq:ratio}. Here, $\Delta p_{\text{err}}$ (see Eq.~\eqref{eq:delta-p-err}) denotes the reduction in error achieved through the use of embezzling catalysts.
    }
    \label{fig:tau-set}
\end{figure*}

In addition to achieving a non-negative catalytic channel capacity with the help of \(\tau^{CS}\), we aim to minimize the consumption of \(\tau^{CS}\) during catalytic quantum communication, as significant changes are undesirable. Specifically, the variation of the embezzling catalyst, measured in terms of purified distance, is upper bounded by
\begin{align}\label{eq:consumption-convex}
        P\left(
           \Tr_{AB}[\mu], \tau^{CS}
         \right)
       \leqslant \sqrt{\frac{2^k}{n}}.
\end{align}
This inequality follows directly from Eq.~\eqref{eq:rho-tau} and the quantum data processing inequality. Consequently, it suggests that increasing the number of copies of $\tau$ used in constructing $\tau^{CS}$ (see Eq.~\eqref{eq:tau-cs}) reduces overall consumption, thereby maintaining the catalytic systems closer to their original state.

The performance of quantum communication systems varies, requiring different error tolerances for different platforms. Our focus is on identifying the minimal dimensional requirements to ensure that the variation in the embezzling catalyst remains within an acceptable error $\delta$. Specifically, for protocols assisted by the convex-split lemma, this issue can be analyzed in terms of the number of copies of the state $\tau$ needed to construct $\tau^{CS}$ (see Eq.~\ref{eq:tau-cs}). According to Eq.~\ref{eq:consumption-convex}, the minimum number of copies, 
$n$, required is determined by
\begin{align}
    n\geqslant \frac{2^k}{\delta^2}.
\end{align}

Since the smallest quantum system is a qubit system, this implies that $m\geqslant 2$. Consequently, Thm.~\ref{thm:csl} encompasses Theorem Thm.~\ref{thm:cqc} as a special case. From Thm.~\ref{thm:csl}, we understand that by allowing slight consumption of the catalysts, specifically utilizing embezzling catalysts, we can reliably transmit more quantum information. However, what's more practical is achieving the same task using embezzling catalysts with lower dimensions. This is crucial considering the current stage of quantum technologies, where creating and manipulating large quantum systems remains a challenge. Further details on reducing the catalytic dimension will be discussed in upcoming sections.

\begin{table*}
    \centering
    \includegraphics[width=1\textwidth]{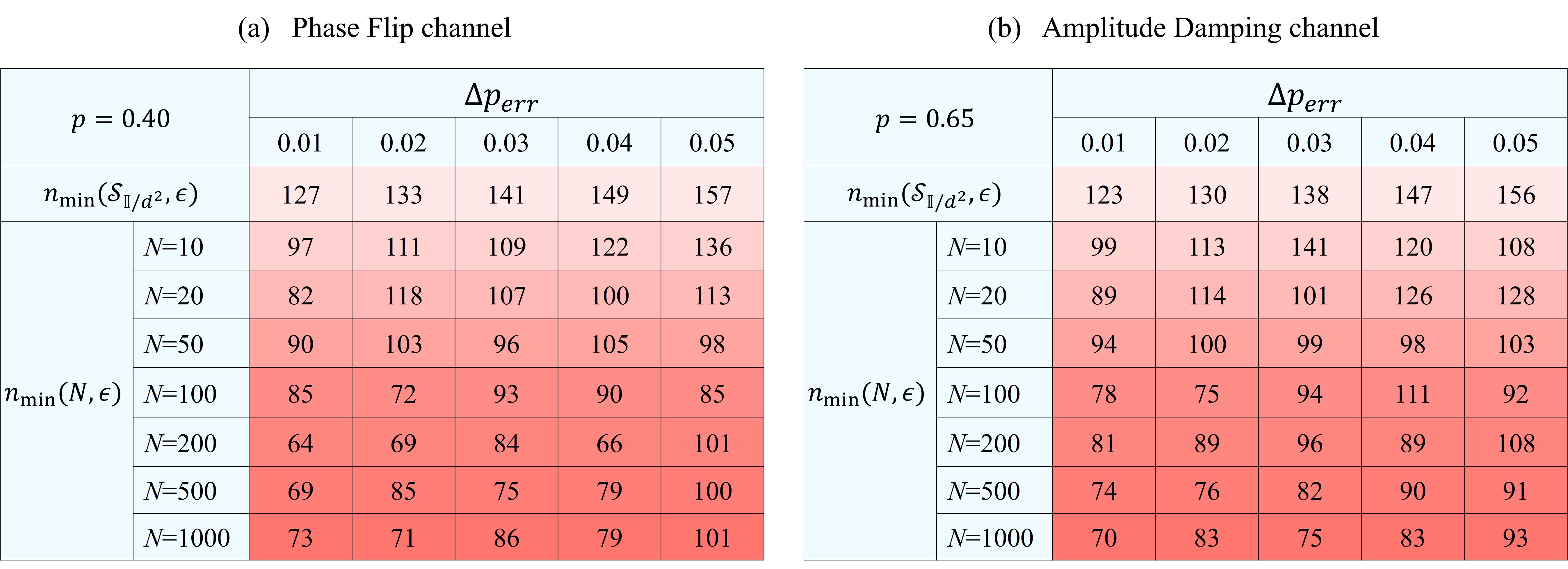}
\caption{{Comparisons between $n_{\text{min}}(\mS_{\mI/d^2}, \epsilon)$ and $n_{\text{min}}(N, \epsilon)$.}
We randomly choose $N$ fully ranked states to construct $\tau$ (see Eq.~\eqref{eq:tau-cs}) for (a) the dephasing channel $\mZ_p$ (with noise parameter $p=0.4$) and (b) the amplitude damping channel $\mN_{\text{AD}}$ (with noise parameter $p=0.65$), and compare the number of copies needed to achieve higher fidelity.
}
\label{tab:my-table}
\end{table*}

\vspace{-1em}
\subsection{Dimensionality Reduction for Catalysts} 
\vspace{-0.5em}

In pursuit of practical catalytic quantum communications, a key consideration is the reduction of the dimensionality of catalyst systems. In this section, our focus is on investigating this challenge to advance catalytic quantum communications. Drawing from the insights of the convex-split lemma, it becomes clear that to create a catalyst enabling quantum communications, the key is to identify quantum states whose support encompasses that of the Choi state $\rho^{\mN}$ (see Eq.~\eqref{eq:choi}) of noisy channel $\mN$ (see Fig.~\ref{fig:cqc}). We denote the collection of all such states as 
   \begin{align}\label{eq:S-supp}
           \mS_{{\rm supp}}
           :=
           \{\tau \,|\, 
           {\rm supp} (\rho^{\mN}) 
           \subseteq 
           {\rm supp} (\tau)\}.
   \end{align}
where $\tau$ is a bipartite quantum state acting on systems $AB$. The minimal dimension $n_{\min}$ required in constructing a state that enables convex-split-lemma-assisted catalytic quantum communication is characterized by iterating over all possible 
$\tau$ from set $\mS_{{\rm supp}}$. Mathematically, it is expressed as
\begin{align}\label{eq:n-min}
    n_{\min}(\mS_{{\rm supp}}, \epsilon):&=\notag\\
    \min\ \  
    &n\notag\\
    \text{s.t.}\ \ 
    &\sqrt{\frac{2^{D_{\text{max}}(\rho^{\mN}\,\|\,\tau)}}{n}} + \sqrt{1-F_U(\tau, \phi^+)}\leqslant 
    \sqrt{\epsilon} ,\notag\\
    &\tau \in \mS_{{\rm supp}}.
\end{align}
Solving the optimization problem outlined in Eq.~\eqref{eq:n-min} poses a considerable challenge due to the mathematical structure of the set $\mS_{{\rm supp}}$ and the presence of non-linear constraints. However, we can explore approaches to bound its performance and investigate strategies for reducing the dimensionality of catalytic systems. One viable avenue involves deriving a computable upper bound for $n_{\min}(\mS_{{\rm supp}}, \epsilon)$ (as defined in Eq.~\eqref{eq:n-min}) by substituting the original $\mS_{{\rm supp}}$ with a carefully chosen subset, thus sidestepping these complexities. In particular, we can explore  states with a fixed structure, such as
\begin{align}\label{eq:S-zeta}
    \mS_{\zeta}:=\{p\phi^++(1-p)\zeta\,|\,p\in[0,1)\}.
\end{align}
Here, $\zeta$ represents a fully-ranked quantum state, ensuring that the support of $\zeta$ covers that of Choi state $\rho^{\mN}$. It is evident that the relationship between $\mS_{{\rm supp}}$ and $\mS_{\zeta}$ can be succinctly expressed through the following inclusion.
\begin{align}
    \mS_{\zeta}\subset\mS_{{\rm supp}}.
\end{align}
Substituting the set $\mS_{{\rm supp}}$ within $n_{\min}(\mS_{{\rm supp}}, \epsilon)$ with $\mS_{\zeta}$ allows us to establish $n_{\text{min}}(\mS_{\zeta}, \epsilon)$. Let $\epsilon$ be fixed; then, they adhere to the following inequality.
\begin{align}
    n_{\text{min}}(\mS_{{\rm supp}}, \epsilon)
    \leqslant
    n_{\text{min}}(\mS_{\zeta}, \epsilon).
\end{align}
It's notable that in many applications of the convex-split lemma, particularly in quantum resource theories, such as catalytic distillation of coherence~\cite{PhysRevA.100.042323}, physicists commonly opt for $\zeta$ to be the maximally mixed state $\mI/d^2$ to derive computable bounds. However, we will demonstrate that selecting alternative full-ranked states could offer advantages in terms of reducing the dimensionality of the catalytic system. To illustrate this point, let's establish $n_{\text{min}}(\mS_{\mI/d^2}, \epsilon)$ as a benchmark and compare the efficacy of our protocol against it.

To initiate our first numerical experiment, we focus on the dephasing channel, namely $\mN(\cdot) = \mZ_p(\cdot)$ (see Eq.~\eqref{eq:dephasing}). We compare the dimension requirements for both the maximally mixed states and four randomly selected fully ranked states, denoted as $\zeta_1$ to $\zeta_4$. Fig.~\ref{fig:tau-set} illustrates that the convex-split-lemma-assisted catalytic quantum communication with these randomly selected fully ranked states requires a smaller dimension for the catalytic system, thereby confirming the following result.
\begin{align}
      n_{\text{min}}(\mS_{\zeta_i}, \epsilon)
      \leqslant 
      n_{\text{min}}(\mS_{\mI/d^2}, \epsilon),
\end{align}
for any $i\in[4]$. To further demonstrate the advantage of using randomly selected fully ranked states, for fixed error threshold $\epsilon$, we define the descent ratio $\vartheta(\epsilon)$ of $n_{\text{min}}(\mS_{\zeta_i}, \epsilon)$ compared with $n_{\text{min}}(\mS_{\mI/d^2}, \epsilon)$ and present the results in Fig.~\ref{fig:tau-set}.
\begin{align}\label{eq:ratio}
    \vartheta(\epsilon)
    :=
    \frac{
    n_{\text{min}}(\mS_{\mI/d^2}, \epsilon)
    -
    n_{\text{min}}(\mS_{\zeta_i}, \epsilon)
    }
    {n_{\text{min}}(\mS_{\mI/d^2}, \epsilon)}.
\end{align}

From the numerical experiment above, we observe that different fully ranked states exhibit varying performance in quantum communication with embezzling catalysts. In Fig.~\ref{fig:tau-set}, $\zeta_4$ requires the smallest dimension for the catalyst across different values of $\epsilon$. In practice, we can randomly select a finite number, $N$, of fully ranked states $\zeta_i$ with $i\in[n]$, choose the optimal one, and denote the minimal copies required for the catalytic system as $n_{\text{min}}(N, \epsilon)$, which satisfies
\begin{align}
n_{\text{min}}(\mS_{{\rm supp}}, \epsilon)
\leqslant
n_{\text{min}}(N, \epsilon).
\end{align}

\begin{figure*}
    \centering
    \includegraphics[width=1\textwidth]{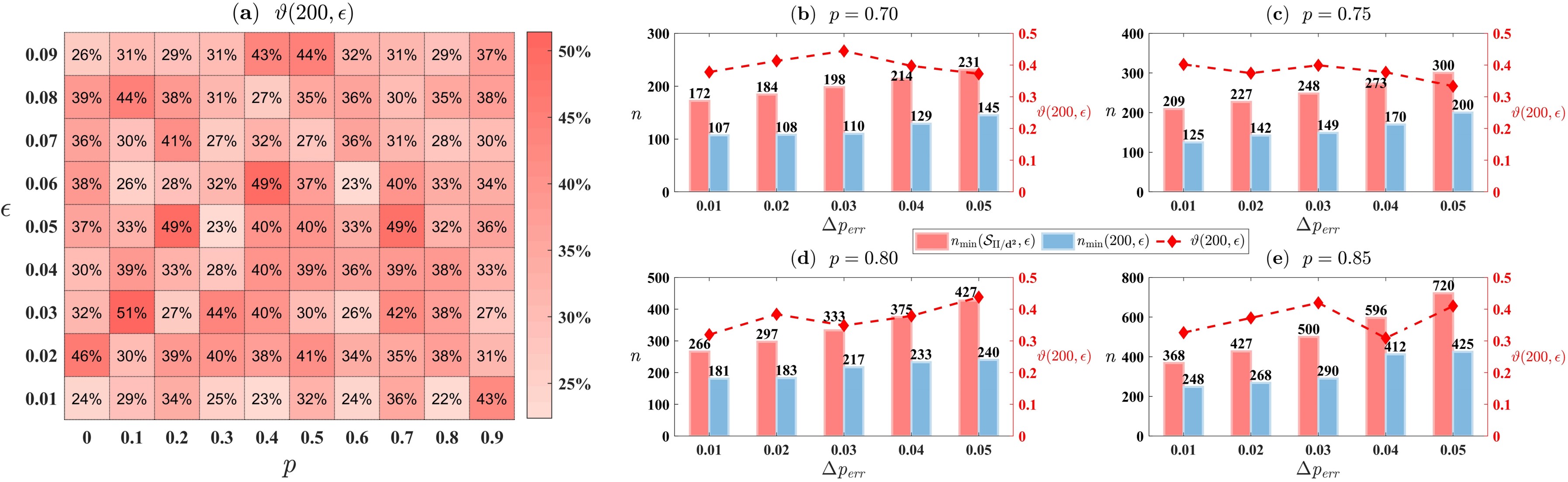}
    \caption{{Comparisons between $n_{\text{min}}(\mS_{\mI/d^2}, \epsilon)$ and $n_{\text{min}}(200, \epsilon)$.}  
    In (a), we present a numerical illustration of the descent ratio $\vartheta(200, \epsilon)$, as defined in Eq.~\eqref{eq:ratio-N}, across different values of $p$ and $\epsilon$.
    In (b), we illustrate the number of copies of $\tau$ (see Eq.~\eqref{eq:tau-cs}) required to achieve improved performance in quantum communication. Additionally, we highlight the advantage of using randomly selected fully ranked states compared to the maximally mixed state. 
    In this context, $\Delta p_{\text{err}}$ as defined in Eq.~\eqref{eq:delta-p-err}, represents the decrement in error attributable to the utilization of embezzling catalysts.
    }
    \label{fig:pauli_channel_theta}
\end{figure*}

For our second example, we use the dephasing channel $\mZ_p$ (see Eq.~\eqref{eq:dephasing}) and the amplitude damping channel $\mN_{\text{AD}}$, defined as 
\begin{align}
    \mN_{\text{AD}}(\cdot):=
    K_{0}\cdot K_{0}^{\dagger}
    +
    K_{1}\cdot K_{1}^{\dagger},
\end{align}
with
\begin{align}
    K_{0}=\begin{bmatrix}1&0\\0&\sqrt{1-p}\end{bmatrix},
    \quad
    K_{1}=\begin{bmatrix}0&\sqrt{p}\\0&0\end{bmatrix},
\end{align}
to investigate the performance of $N$ randomly selected fully ranked states. We compare their performance against that of the maximally mixed state $\mI/d^2$ and present the results in TABLE~\ref{tab:my-table}. In principle, the performance should improve as more fully ranked states are chosen. However, our numerical experiments indicate that beyond a certain number of samples ($200$ for the dephasing channel and $100$ for the amplitude damping channel), the benefits of increasing the number become negligible. Compared with $n_{\text{min}}(\mS_{\mI/d^2}, \epsilon)$, the descent ratio $\vartheta(N, \epsilon)$ achieved by $N$ fully ranked states is defined as follows 
\begin{align}\label{eq:ratio-N}
    \vartheta(N, \epsilon)
    :=
    \frac{
    n_{\text{min}}(\mS_{\mI/d^2}, \epsilon)
    -
    n_{\text{min}}(N, \epsilon)
    }
    {n_{\text{min}}(\mS_{\mI/d^2}, \epsilon)}.
\end{align}
For the dephasing channel $\mZ_p$ with $200$ fully ranked states, we conducted additional numerical experiments, and the results are presented in Fig.~\ref{fig:pauli_channel_theta}. These results further confirm the advantage of using randomly selected fully ranked states over the maximally mixed state.

\subsection{Embezzling-State-Assisted Quantum Communication}
Our next catalytic quantum communication protocol uses the {\it embezzling states}, denoted as $\tau^{E}$. Originally introduced by van Dam and Hayden in the context of entanglement theory~\cite{PhysRevA.67.060302}, these states have proven to be remarkably versatile, extending their utility across various quantum resources. They have found applications in thermodynamics~\cite{GOUR20151,Ng_2015}, coherence~\cite{PhysRevA.100.042323,PhysRevLett.113.150402}, and beyond~\cite{10.1063/1.4938052,7377103}, enriching our understanding of these fields. Furthermore, embezzling states have also played an important role in the quantum reverse Shannon theorem~\cite{4957651,doi:10.1007/s00220-011-1309-7}. We now proceed to a detailed exploration of embezzling-state-assisted catalytic quantum communication.

\begin{theorem}
[{\bf (Embezzling-State-Assisted Catalytic Channel Capacity)}]
\label{Thm:es-ccc}
Given a noisy quantum channel $\mN$ (see Fig.~\ref{fig:cqc}) and an error threshold $\epsilon>0$, we can utilize the embezzling state $\tau^{E}$ to improve quantum communication
\begin{equation}\label{eq:catalyst-embezzling}
        \ket{\tau^{E}} = 
        \frac{1}{\sqrt{c_M}}\sum_{j=1}^M\frac{1}{\sqrt{j}}\ket{jj} ,
\end{equation}
where the coefficient $c_M:=\sum_{j=1}^{M}(1/j)$ normalizes the state $\tau^{E}$.
If the Schmidt rank of the embezzling state $\tau^{E}$ satisfies
\begin{align}\label{eq:Schmidt-rank}
        M
        \geqslant 
        \lceil m^{(1-\sqrt{1-\epsilon})^{-1}} \rceil,
\end{align}
 then there exists a catalytic quantum communication protocol that satisfies
\begin{align}\label{eq:capacity2}
	Q_c^{\epsilon}(\mN)
        \geqslant 
        \log m.
\end{align}
\end{theorem}

\begin{figure}
	\begin{center}   
    \includegraphics[width=0.48\textwidth]{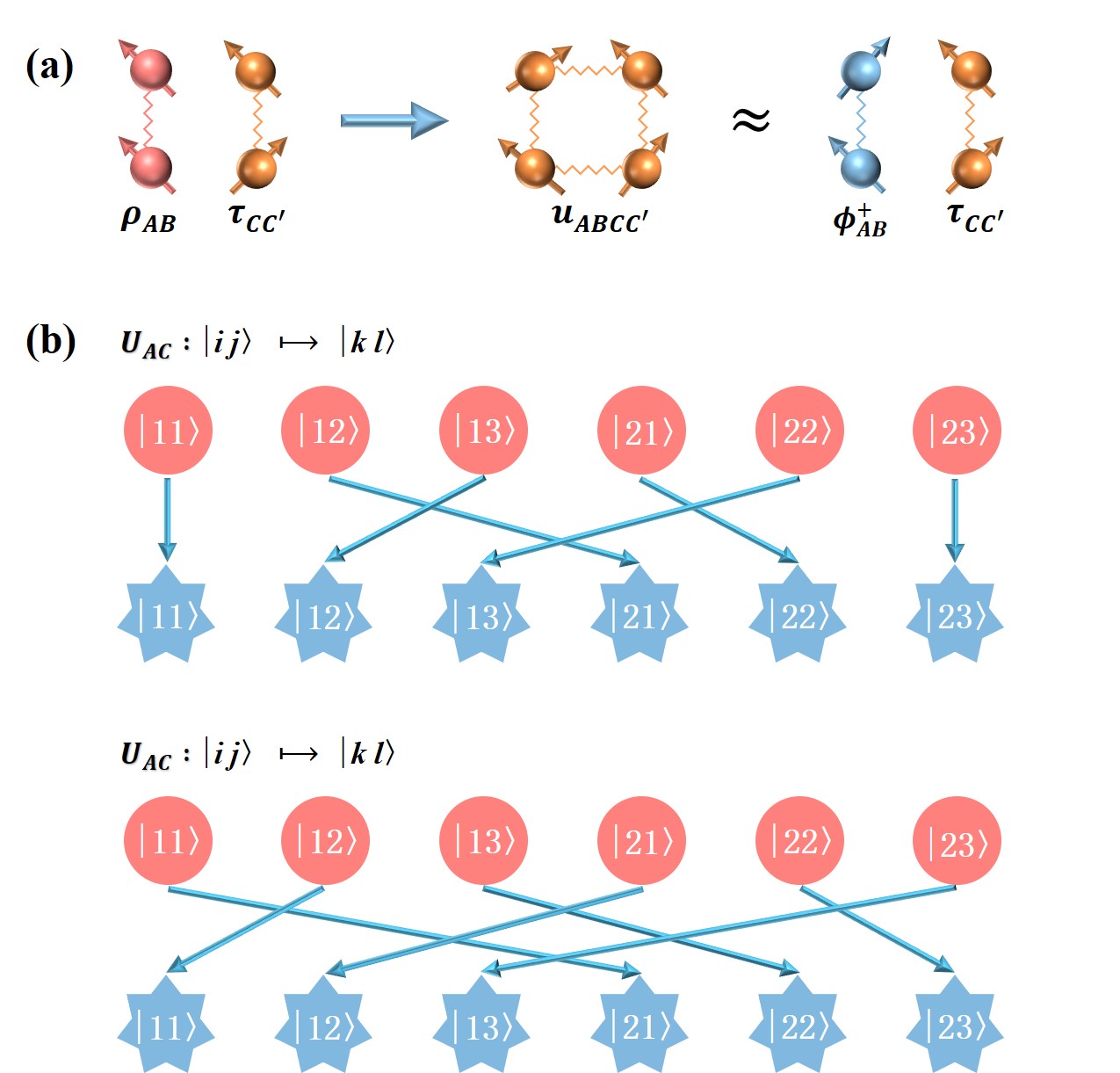}
	\end{center}
	\caption{{Embezzling-state-assisted protocol.} 
(a). After receiving $\rho_{AB}$ through the implementation of the noisy quantum channel $\mN$, we perform some LOCC operations that transform $\rho_{AB}\otimes\tau_{CC^{'}}$ into $\mu_{ABCC^{'}}$, which can be used to approximate $\phi^+_m\otimes \tau^E$.
(b). We offer two distinct approaches to realizing the unitary operation $U_{AC}$, utilized to derive $\mu$ and transform $\omega$ into $\phi^+_m\otimes \tau^E$, and illustrate the process for the case when $m=2$ and $M=3$.
} 
\label{fig:ES}
\end{figure}

\begin{proof}  
In the context of embezzling-state-assisted catalytic channel communication, the communication power predominantly stems from the embezzling catalyst. Specifically, within this protocol, the necessity of pre-processing becomes negligible. Regardless of the bipartite state obtained on systems $AB$ through the implementation of noisy channel $\mN$, it can be effectively substituted with $\ket{11}$ via LOCC. Subsequently, a joint unitary operation $U_{AC} \otimes U_{BC^{'}}$ is enacted upon the entire system $\ketbra{11}\otimes \tau^{E}$, leading to
\begin{align}\label{eq:mu}
    \mu=
    U_{AC} \otimes U_{BC'}
    (\ketbra{11}{11}_{AB}\otimes\tau^E_{CC^{'}}) 
    U_{AC}^{\dagger} \otimes U_{BC'}^{\dagger}.
\end{align}
Here the unitary map $U_{AC}$ is defined as
\begin{align}\label{eq:unitary}
       U_{AC}: 
       \mH_{A}\otimes \mH_{C} &\rightarrow \mH_{A}\otimes \mH_{C}\notag\\
       \ket{i}_A \otimes \ket{j}_C &\mapsto \ket{k}_A \otimes \ket{l}_C,
\end{align}
where $l=\lceil\frac{(i-1)M+j}{m}\rceil$, and $k=(i-1)M+j-(l-1)m$. 

To assess the performance of the quantum state $\mu$ in approximating $\phi^+_m\otimes \tau^E$, we introduce a supplementary quantum state $\ket{\omega}$ to facilitate our analysis
\begin{equation}
     \ket{\omega}
     =
     \sum^{m}_{i=1} \sum^{M}_{j=1}\omega_{ij} \ket{ii}\ket{jj}.
\end{equation}
The coefficient $\omega_{ij}:=1/\sqrt{\lceil (i-1)M+j/m \rceil mc_M}$ are arranged in the dictionary order, namely
\begin{align}
    \omega_{11}\geqslant\cdots\geqslant  \omega_{1M}\geqslant\cdots\geqslant\omega_{m1}\geqslant\cdots\geqslant\omega_{mM}.
\end{align}
For the case \(i=1\), we have
\begin{equation}\label{eq:omega_j}
     \omega_{1j}
     =
     1/{\sqrt{\lceil j/m \rceil mc_M}} \leqslant 1/{\sqrt{jc_M}}.
\end{equation}

\noindent Hence, the inner product between states $\ket{11}_m\otimes\ket{\tau^{E}}$ and $\ket{\omega}$ is bounded by
\begin{align}
    \bigg|\bigg\langle
    \ket{11}\otimes\ket{\tau^E} ,
    \ket{\omega}
    \bigg\rangle\bigg|
    \geqslant 
    \frac{\log M- \log m}{\log M}. \label{eq:omega_Fu}
\end{align}
Another important property of $\omega$ is that by applying the joint unitary operation $U_{AC} \otimes U_{BC^{'}}$, $\omega$ will be transformed into $\phi^+_m\otimes \tau^E$, i.e.,
\begin{align}\label{eq:Lambda-eU}
    U_{AC} \otimes U_{BC^{'}}\ket{\omega}
    =
    \ket{\phi^+_m}\otimes\ket{\tau^E} .
\end{align}

\begin{figure*}
	\begin{center}   
    \includegraphics[width=1\textwidth]{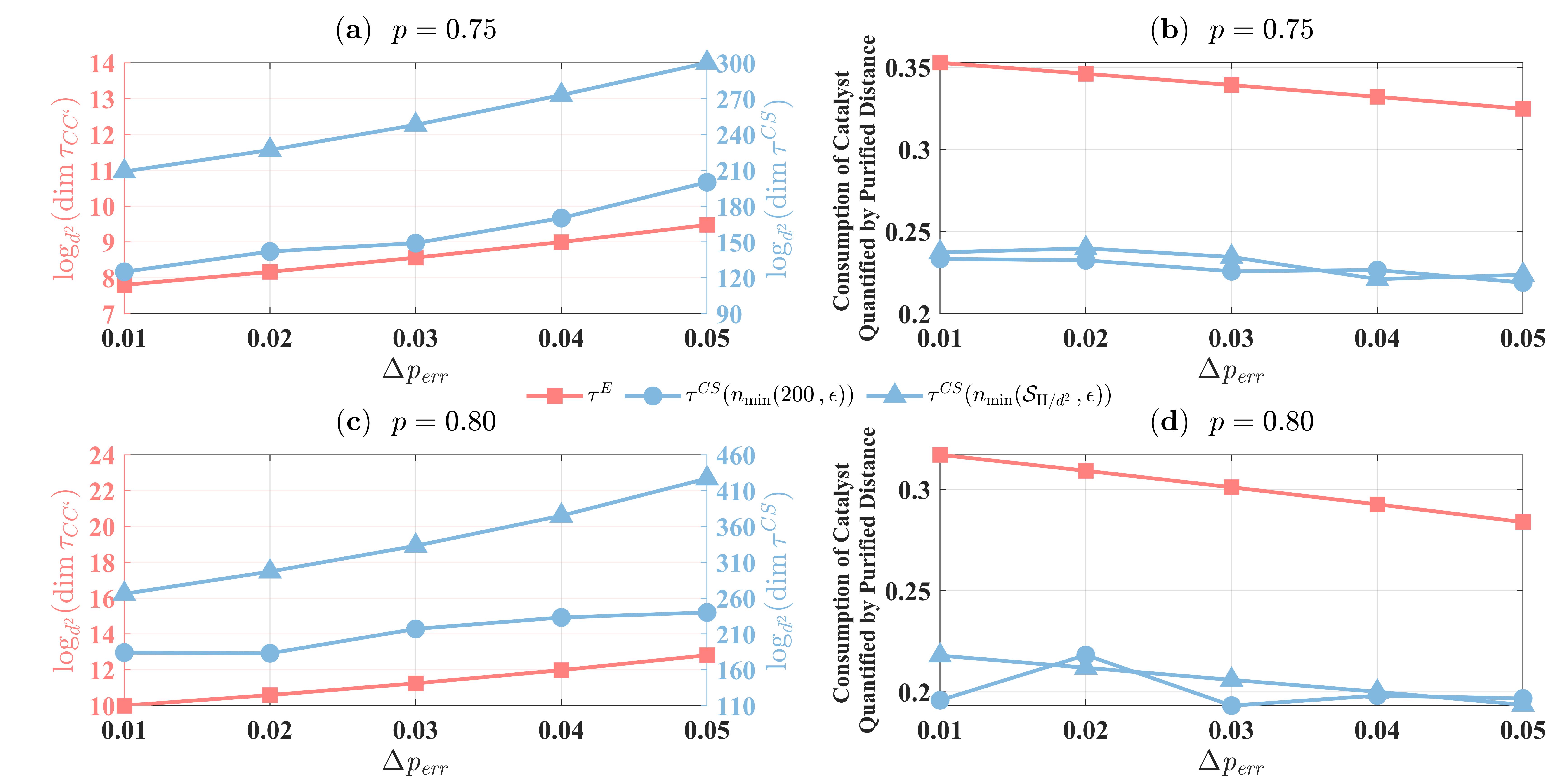}
	\end{center}
\caption{{Comparisons across different catalytic quantum communication protocols.} 
Figures (a) and (c) illustrate the catalyst size needed for achieving performance improvements across different catalytic quantum communication protocols within dephasing channel $\mZ_p$ (see Eq.~\eqref{eq:dephasing}), with noise parameters $p=0.75$ and $p=0.80$. In these figures, $\tau^E$ represents the catalyst dimension based on embezzling states (see Eq.~\eqref{eq:catalyst-embezzling}). In comparison, $\tau^{CS}(n_{\rm min}(200\,, \epsilon))$ and $\tau^{CS}(n_{\rm min}(\mS_{\mI/d^2}\,, \epsilon))$ indicate the catalyst dimensions constructed using the convex-split lemma (see Lem.~\ref{lem:CS}). The former method involves selecting $200$ randomly chosen full-rank states for constructing $\tau$ (see Thm.~\ref{thm:csl}), whereas the latter employs maximally mixed states $\mI/d^2$. Figures (b) and (d) present the usage of embezzling catalysts in communication protocols, evaluated by purified distance. The blue line delineates the consumption for the embezzling catalyst $\tau^{CS}$, as outlined in Eq.~\eqref{eq:consumption-convex}. Meanwhile, the pink line details the usage of the embezzling state $\tau^E$, as specified in Eq.~\eqref{eq:consumption-emb}. These comparisons indicate that embezzling-state protocols surpass convex-split-lemma counterparts in performance at the same catalyst scale, yet they undergo greater changes from their initial state after catalytic quantum communication.
} 
\label{fig:06}
\end{figure*}

With all the preparatory work complete, let us now evaluate the performance of $\mu$ (see Eq.~\eqref{eq:mu}) in approximating $\phi^+_m\otimes \tau^E$, which is characterized by the following inequality
\begin{align}\label{eq:fidelity-embezzling}
F_U\left(\mu, \phi^+_m\otimes \tau^E \right)
=
& F_U\left(\ketbra{11}{11}\otimes\tau^E,\omega\right)\notag\\
=
&\bigg|\bigg\langle
\ket{11}_m\otimes\ket{\tau^E} ,
\ket{\omega}
\bigg\rangle\bigg|^2.
\end{align}
Now the performance of this catalytic quantum communication, measured by its error $p_{\text{err}}^c(\mu)$ (see Eq.~\eqref{eq:capacity}), satisfies 
\begin{align}\label{eq:perr-c}
   p_{err}^c(\mu)
   = 
   1-F_U(\mu, \phi^+_m\otimes \tau^{E})
   \leqslant 
   \epsilon.
\end{align}
In other words, the catalytic channel capacity of the noisy channel $\mN$ is lower bounded by $m$; that is
\begin{align}
	Q_c^{\epsilon}(\mN)
        \geqslant 
        m,
\end{align}
which provides an alternative proof of Thm.~\ref{thm:cqc}.
\end{proof}

Beyond showcasing the enhanced efficacy of the protocol facilitated by embezzling states, our analysis extends to the quantification of embezzling catalyst consumption during catalytic quantum communication. We particularly investigate the alteration in the embezzling catalysts by employing the purified distance metric. To initiate this examination, we consider the resultant state of the catalytic quantum communication, denoted as follows
\begin{align}
    \xi^{E}
    =
    \frac{1}{c_M}\sum^M_{t=1}
         \bigg(\sum^{K-1}_{i=1}\frac{1}{\sqrt{k_it}}\left(\ketbra{ii}{KK}+\ketbra{KK}{ii}\right)&\notag\\
        \quad\quad\quad 
        + \frac{1}{t}\ketbra{KK}{KK}&\bigg),
\end{align}
where $K:=\lceil t/m \rceil$, and $k_i:= t-\lfloor (t-1)/m \rfloor m +(i-1)m$. The degree to which the embezzling state is consumed during catalytic quantum communication is quantifiable by its variation, expressed as follows
\begin{align}
        F_{U}(\xi^{E}, \tau^E)
        = 
        \frac{1}{c^2_M}\sum^M_{t=1}
         \bigg(\sum^{K-1}_{i=1}\frac{2}{\sqrt{ik_itK}}
        + \frac{1}{tK}\bigg).
\end{align}
This yields the precise formulation of the purified distance between the states $\xi^{E}$ and $\tau^E$, specifically
\begin{align}\label{eq:consumption-emb}
        P\left(\xi^{E}, \tau^{E}\right)
        =
        \sqrt{1-\frac{1}{c^2_M}\sum^M_{t=1}
         \bigg(\sum^{K-1}_{i=1}\frac{2}{\sqrt{ik_itK}}
        + \frac{1}{tK}\bigg)}.
\end{align}
Employing Eqs.~\eqref{eq:omega_Fu} and \eqref{eq:fidelity-embezzling} in conjunction with the quantum data processing inequality enables us to deduce a more streamlined upper bound,
\begin{align}
    P\left(\xi^{E}, \tau^{E}\right)
    \leqslant
    \sqrt{2\log_{M} m},
\end{align}
where $M$ represents the Schmidt rank of the embezzling catalyst $\tau^E$.

In embezzling-state-assisted quantum communication, the minimal dimension of the embezzling state, or equivalently, the minimum Schmidt rank $M$, required to ensure that the variation in the embezzling catalyst remains within an acceptable error margin $\delta$ during catalytic quantum communication is determined by
\begin{align}
    M\geqslant m^{\frac{2}{\delta^2}}.
\end{align}

It is worth noting that the unitary operation $U_{AC}$ used to obtain $\mu$ and transform $\omega$ into $\phi^+_m\otimes \tau^E$ is not unique. Another feasible unitary map, along with a schematic diagram of the embezzling-state-assisted protocol, is shown in Fig.~\ref{fig:ES}. So far, we have derived several protocols for catalytic quantum communication. Specifically, for the dephasing channel $\mZ_p$ (see Eq.~\eqref{eq:dephasing}), we compare their performance and outline the dimension requirements for both the convex-split-lemma-assisted and embezzling-state-assisted protocols in Fig.~\ref{fig:06}. In this numerical investigation, we evaluated the efficacy of various protocols by examining the dimensional requirements for catalytic systems and observing the alterations in embezzling catalysts during catalytic quantum communication, quantified by purified distance. Protocols assisted by the convex-split lemma that utilize randomly selected full-rank states to construct $\tau$ (see Eq.~\eqref{eq:tau-cs}) exhibit superior performance over those employing maximally mixed states $\mI/d^2$. Nonetheless, when considering identical dimensions of catalytic systems, the embezzling-state-assisted protocol outperforms convex-split-lemma-assisted protocols. This enhancement in performance is accompanied by a more pronounced variation in catalysts throughout the catalytic quantum communication process.

\begin{figure}
	\begin{center}  
\includegraphics[width=0.47\textwidth]{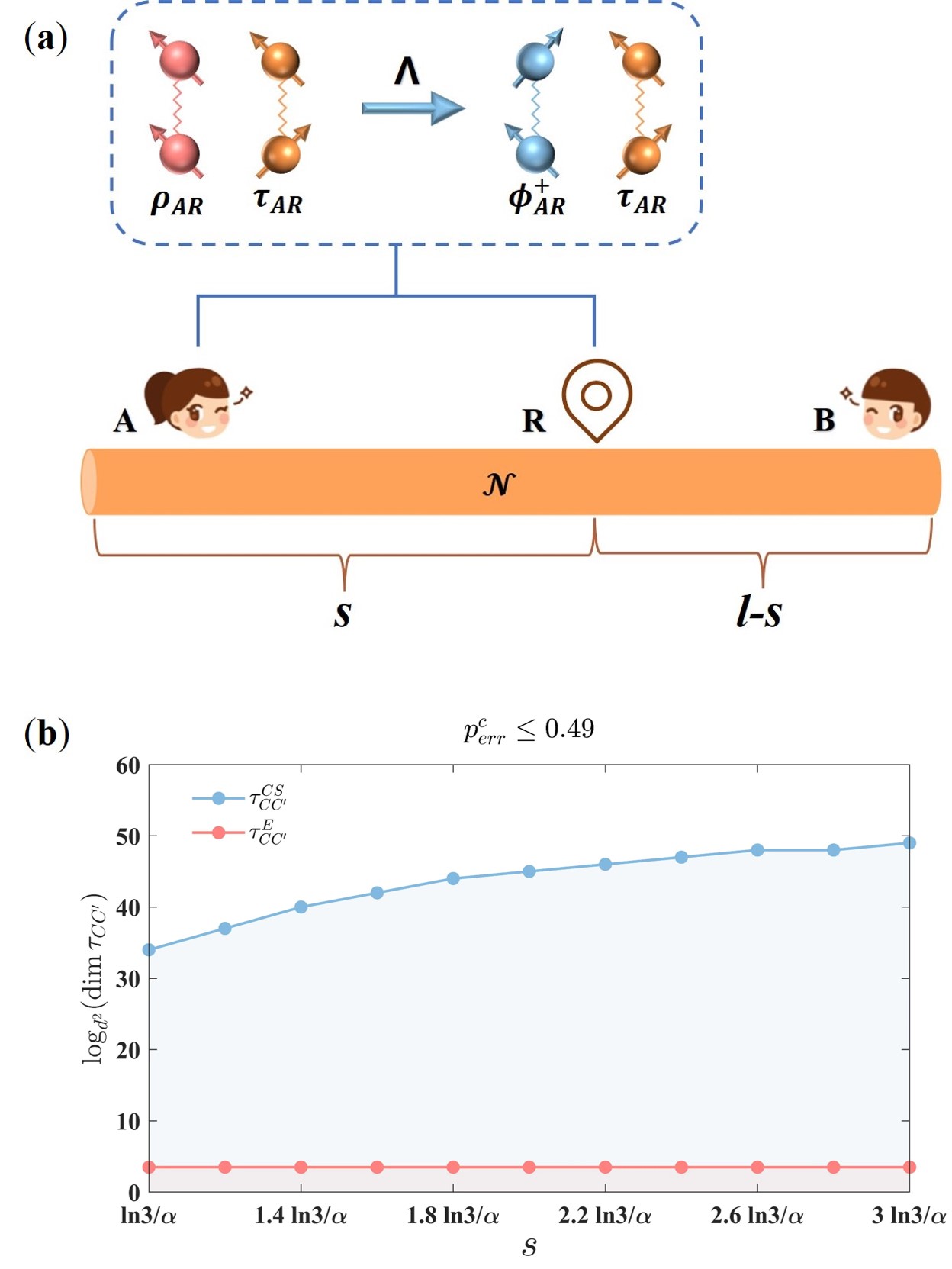}
\end{center}
\caption{{Entanglement distribution.} 
(a) Schematic diagram of entanglement distribution from Alice to Bob through a depolarizing channel (see Eq.~\eqref{eq:depolarizing}), with an intermediate node located at a distance $s$ from Alice and the total distance between Alice and Bob being $l$.
(b) Comparison of the dimension requirements for convex-split-lemma-assisted and embezzling-assisted entanglement distribution protocols.
}
\label{fig:07}
\end{figure}

\subsection{Lower Bound for Catalytic Channel Capacity}
An important feature of catalytic channel capacity $Q_c^{\epsilon}(\mN)$ is that, with the use of embezzling catalysts, we can establish an analytical lower bound that is independent of the specific noisy channel $\mN$. This bound exists due to the presence of universal embezzling catalysts, such as embezzling states $\tau^{E}$ (see Eq.~\eqref{eq:catalyst-embezzling}). Specifically, an analytical lower bound for catalytic channel capacity is described by the following theorem.

\begin{corollary}
[{\bf (Lower Bound for Catalytic Channel Capacity)}]
\label{Thm4}
Given a noisy quantum channel $\mN$ (see Fig.~\ref{fig:cqc}) and an error threshold $\epsilon>0$, if a catalyst of dimension $d_c$ is allowed, then the catalytic channel capacity is lower bounded by
\begin{align}
	Q_c^{\epsilon}(\mN)
        \geqslant 
        (1-\sqrt{1-\epsilon})\log (d_c-1).
\end{align}
\end{corollary}
\begin{proof}
To prove this result, we employ a protocol that precisely attains the lower bound. Specifically, we employ the embezzling-state-assisted quantum communication described in Thm.~\ref{Thm:es-ccc}, where $\log d$ qubits can be reliably transmitted when $d_c$ takes the following value
\begin{align}
        d_c
        = 
        \lceil d^{(1-\sqrt{1-\epsilon})^{-1}} \rceil,
\end{align}
which immediately implies that
\begin{align}
    \log d
    \geqslant
    (1-\sqrt{1-\epsilon}) \log (d_c-1).
\end{align}
This completes the proof.
\end{proof}

Here, we utilize a specific protocol to establish the lower bound of catalytic channel capacity $Q_c^{\epsilon}(\mN)$. However, we have yet to determine how to formulate the optimal bound that is independent of the specific noisy channel $\mN$. Intuitively, we anticipate that with more information or by allowing the bound to be channel-dependent, a better lower bound could be achieved. Conversely, besides the lower bound presented here, we are also interested in the upper bound for catalytic channel capacity $Q_c^{\epsilon}(\mN)$, which would limit the improvement achievable through the use of catalysts. Unfortunately, none of these questions have straightforward solutions and are complicated by the optimization over LOCC superchannels. Therefore, we leave them for future explorations.

\subsection{Long-Distance Entanglement Distribution}
Long-distance entanglement distribution is a crucial aspect of quantum communication, aiming to establish entangled states between distant parties. This process is fundamental for quantum networks, quantum cryptography, and distributed quantum computing. Ensuring reliable entanglement over large distances involves overcoming significant challenges, including decoherence, photon loss, and other noise effects. In conventional approaches, part of the entangled state is transmitted through a noisy quantum channel, which degrades the quality of the entangled state. The use of correlated catalysts can enhance the performance of entanglement distribution, extending the distance over which high-fidelity entanglement can be maintained~\cite{Datta2024entanglement}. We propose that employing embezzling catalysts can further improve this process, surpassing the limitations of correlated catalysts and enabling even longer-distance entanglement distribution.

In entanglement distribution, Alice aims to share entanglement with a distant party, Bob. To achieve this, she prepares a maximally entangled state of qubit systems, $\phi^+_2$, and sends part of it to Bob through a noisy channel, such as a depolarizing channel $\mN_l$, 
\begin{align} \label{eq:depolarizing}
	\mN_l(\cdot)
        =
        e^{-\alpha l}(\cdot)+(1-e^{-\alpha l})\frac{\mI}{2},
\end{align}
where $ l$ is the length of the channel and $\alpha\geqslant 0$ represents the damping parameter. As the distance of distribution increases, the entanglement weakens. After transmitting the state over a distance $s$ to a point $R$, the resulting bipartite state shared between points $A$ and $R$ is denoted as $\rho_{AR}$, as illustrated in Fig.~\hyperref[fig:07]{7(a)}.

In the conventional approach, where no auxiliary system is considered, the longest distance over which Alice can distribute entanglement is $(\ln 3)/\alpha$. This means that if Bob is located at a distance of at least $(\ln 3)/\alpha$ from Alice, she cannot share entanglement with him after sending part of the system through the depolarizing channel. In the correlated-catalyst-assisted protocol, if a correlated catalyst is employed at the point $R$ with $s < l/2 < (\ln 3)/\alpha$, then for Bob located within $l<(2\ln 3)/\alpha$, Alice can still distribute entanglement via the depolarizing channel~\cite{Datta2024entanglement}. However, if Alice and Bob are separated by a distance greater than $(2\ln 3)/\alpha$, entanglement distribution will fail, even with the help of correlated catalysts. With embezzling catalysts, we can surpass this limit and achieve longer entanglement distribution. To be more specific, employing either the convex-split-lemma-assisted protocol, which is based on $200$ randomly generated full-ranked states, or the embezzling-state-assisted protocol, allows Alice to establish entanglement with Bob across spans exceeding $(2\ln 3)/\alpha$. It is worth mentioning that the dimension of the catalysts varies depending on the intermediate node's distance from Alice. Numerical experiments presented in Fig.~\hyperref[fig:07]{7(b)} compare the performance of these two different catalytic protocols. Notably, the embezzling-state-assisted protocol stands out for achieving equivalent performance with fewer dimensions in the catalytic systems, thus indicating superior efficiency.


\section{Catalytic Classion Information Transmission}\label{SEC:IV}
Until now, our protocols, along with those discussed in existing literature, have focused on enhancing the transmission of quantum information using catalysts. Yet, we are captivated by the following question: Can catalysts also improve the efficiency of transmitting classical information? The answer is affirmative. Here, we explore superdense coding, demonstrating how catalysts can improve its performance. By doing so, we address a previously overlooked aspect of catalytic communication, underscoring that quantum catalysts hold potential to enhance the transmission of both quantum and classical information.

\begin{figure}[t]
\begin{center}  
\includegraphics[width=0.47\textwidth]{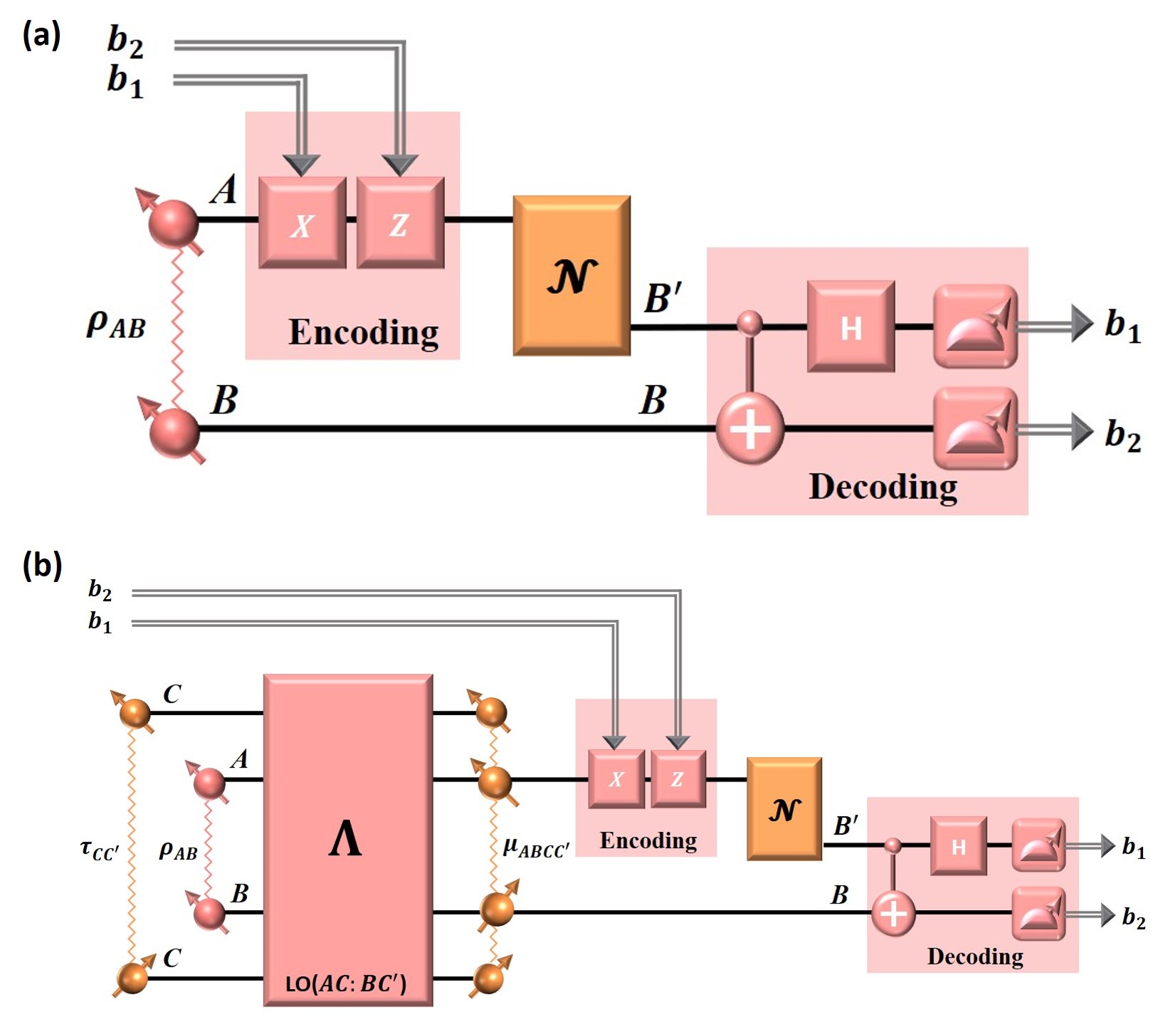}
\end{center}
\caption{{Superdense coding.}
In standard superdense coding (a), Alice encodes the classical information $(b_1, b_2)$ using the local gate $Z^{b_2}X^{b_1}$ and sends her qubit to Bob. Assisted by the entangled state $\rho_{AB}$ shared between them, Bob can decode the classical information via Bell measurements.
In catalytic superdense coding (b), Alice and Bob utilize an additional catalyst $\tau_{CC^{'}}$. After applying local operations, the catalytic systems remain almost unchanged, but the quality of entanglement is enhanced, leading to improved performance of superdense coding.
}
\label{fig:sdc}
\end{figure}

Superdense coding exemplifies the use of quantum resources, specifically quantum entanglement, to transmit classical information~\cite{651037,PhysRevLett.83.3081,1035117,10.1063/1.1495877}, where two bits of classical information are transmitted using a pre-shared entangled state $[qq]$ along with a single application of a noiseless qubit channel $[q\to q]$. The process can be succinctly summarized as the following inequality
\begin{align}
       \langle qq\rangle + [q\to q] 
       \geqslant 
       2[ c \to c].    
\end{align}
Consider the qubit scenario as an illustration: Alice and Bob initially share a maximally entangled state. Alice then applies one of four possible Pauli gates to her qubit. Subsequently, Alice transmits her particle to Bob through a noiseless quantum channel. Upon reception, Bob performs Bell measurements. By examining the classical outcome of these measurements, Bob can read out the message transmitted by Alice. A schematic diagram illustrating this process is depicted in Fig.~\hyperref[fig:sdc]{8(a)}.

In realistic superdense coding protocols, generic entangled states are typically employed, leading to inevitable noise in the state transmitted by Alice. Instead of the ideal maximally entangled state $\langle qq\rangle$, a noisy bipartite state $\rho_{AB}$ is shared between Alice and Bob, as depicted in Fig.~\hyperref[fig:sdc]{8(b)}. Unlike other catalytic quantum communication protocols, such as teleportation or others, where LOCC are assumed to be cost-free, the situation here is different. In superdense coding, the objective is to simulate classical communications. Therefore, direct classical communication between Alice and Bob should not be permitted; otherwise, superdense coding would become trivialized. Consequently, the set of allowed operations should be limited to collections of all local operations, i.e., LO. For correlated catalysts, well-known options like Duan's state cannot be used for catalytic superdense coding. Similarly, for embezzling catalysts, protocols like the convex-split-lemma-assisted protocol cannot be employed due to the requirement of LOCC. However, the embezzling-state-assisted protocol remains viable, as only local operations are necessary. This allows us to employ local operations $\mE\in\text{LO}(AC:BC^{'})$ on $\rho_{AB}\otimes\tau^{E}_{CC^{'}}$ to enhance the quality of entanglement shared between Alice and Bob, thereby further improving the performance of superdense coding. For further insight, please refer to Eq.~\eqref{eq:catalyst-embezzling}, which details the construction of the embezzling state. It's worth noting that while correlated catalysts like Duan's state cannot be used for superdense coding, the potential existence of other correlated catalysts that could enhance superdense coding remains an open question.

For a noisy entangled state $\rho_{AB}$, the effectiveness of superdense coding is assessed through its superdense coding capacity $C(\rho_{AB})$, which is defined as follows~\cite{TohyaHiroshima_2001,PhysRevLett.93.210501}
\begin{align}\label{eq:classcapacity}
    C(\rho_{AB}) = \log d - H(A|B),    
\end{align}
where $H(A|B):= H(\rho_{AB})- H(\rho_{B})$ stands for the quantum conditional entropy, with $H$ being the von Neumann entropy for quantum states. Operator $\rho_{B}$ represents the reduce state of $\rho_{AB}$, i.e., $\rho_{B}:= \Tr_{A}[\rho_{AB}]$. For simplicity, we assume that the dimension of system $A$ equals the dimension of system $B$, denoted by $d$. After undergoing the catalytic process (see Fig.~\hyperref[fig:ES]{5(a)}) involving the embezzling state $\tau^{E}$ (see Eq.~\eqref{eq:catalyst-embezzling}), the state acting on systems $AB$ transforms to
\begin{align}
    \rho^E_{AB}:&= \Tr_{CC'}[\mu_{ABCC'}]\notag\\
    &= \frac{1}{c_M}\sum^d_{\stackrel{{\scriptstyle i=1}}{j=1}}\sum^{K}_{k=0}\frac{1}{\sqrt{(i+kd)(j+kd)}}\ketbra{ii}{jj}_{AB},
\end{align}
where $K:= \lfloor (M-\max \{i,j\})/d \rfloor$, and $M$ represents the Schmidt rank of the embezzling state, which affects the performance of catalytic superdense coding. Numerical experiments demonstrating the enhancement of superdense coding capacity through the use of embezzling states with different Schmidt ranks $M$ are shown in Fig.~\ref{fig:superdense-capacity}. The results indicate that when the dimension of the embezzling state is sufficiently large, the superdense coding capacity approaches $2\log d$ (see Fig.~\hyperref[fig:superdense-capacity]{9(b)}). In low-dimensional systems, improvements are also observed; however, this comes at the cost of significant changes to the catalytic systems (see Fig.~\hyperref[fig:superdense-capacity]{9(a)}).

\begin{figure}
    \centering
    \includegraphics[width=0.47\textwidth]{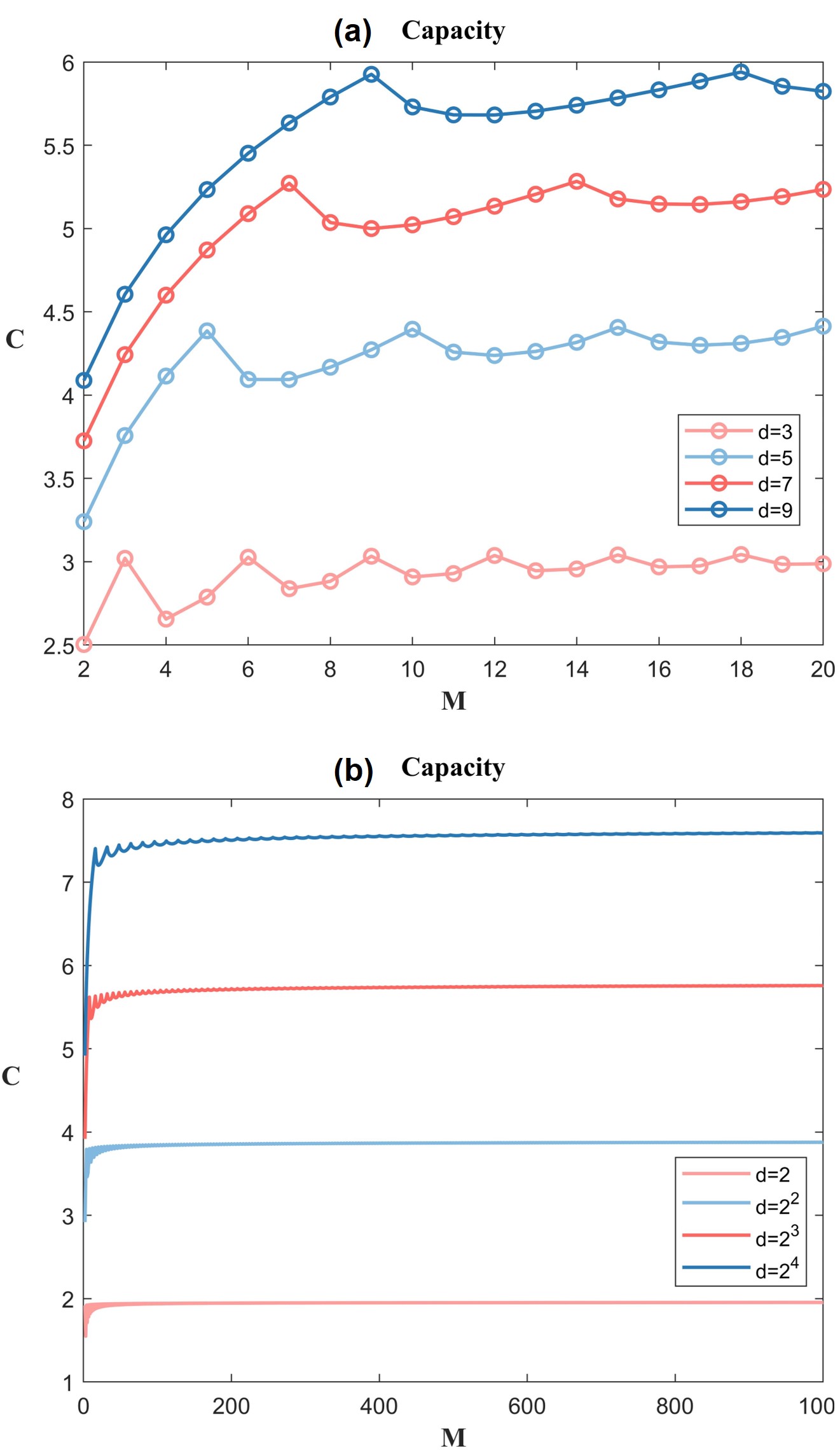}
    \caption{{Catalytic superdense coding capacity.}  
    In (a), we present a numerical illustration of the catalytic superdense coding capacity, as defined in Eq.~\eqref{eq:classcapacity}, for various values of $M$ and $d$. In (b), as the dimension increases, the capacity converges towards the theoretical upper bound of $2\log d$.}
    \label{fig:superdense-capacity}
\end{figure}

\section{Discussion}\label{sec:conclusion}
In this work, we have introduced communication with quantum catalysts, specifically focusing on an additional entangled resource known as an embezzling catalyst. This resource enhances the catalytic capacity of quantum communication through noisy channels and optimizes the performance of superdense coding. Unlike exact catalysts, embezzling catalysts experience only negligible changes during the communication process, allowing them to be recycled for subsequent use. By easing the constraint of preserving catalysts throughout information processing, we can develop universal catalysts. These catalysts enable the transmission of quantum information regardless of the noisy channel conditions in quantum communication and facilitate superdense coding irrespective of the entangled state shared between the sender and receiver in classical communication. To demonstrate our scheme, we constructed embezzling catalysts using two distinct methods: the convex-split lemma~\cite{PhysRevLett.119.120506} and the embezzling state~\cite{PhysRevA.67.060302}. For quantum information transmission, both methods successfully achieve a non-zero catalytic channel capacity with finite-dimensional embezzling catalysts. For classical information transmission, only the embezzling-state-assisted protocol can be employed, as direct classical communication between parties is not allowed.

The phenomenon of embezzling in quantum information theory refers to the remarkable ability to extract resources, such as entanglement, from a reference state without significantly altering it. This is akin to taking a cup of water from the ocean and leaving the ocean nearly unchanged. However, the dimensionality of the reference state, or catalyst, is crucial for this process. High-dimensional systems pose challenges in preparing, distributing, and manipulating entangled states precisely. Consequently, low-dimensional catalysts are necessary. In this work, we explore methods to reduce the dimensionality of embezzling catalysts using the convex-split lemma. We demonstrate the dimensional requirements for both convex-split-lemma-assisted and embezzling-state-assisted protocols in quantum communication. Detailed numerical experiments are also provided to support our theoretical results.

Catalytic channel capacity is a critical metric for assessing the capability of transmitting quantum information through a noisy quantum channel. Superdense coding, on the other hand, is a fundamental information task that illustrates how entanglement can enhance the transmission of classical information. Here, we have demonstrated that embezzling catalysts can enhance both quantum and classical information transmission. In quantum communications, correlated catalysts can improve catalytic channel capacity. However, it remains an open question whether similar enhancements can be achieved for superdense coding. The conventional approach, such as using Duan's state, does not apply to superdense coding because classical communication between parties is not permitted. Beyond the tasks investigated here, can embezzling catalysts enhance other communication tasks? The answer is affirmative. In our accompanying paper~\cite{CP}, we explore their roles in quantum teleportation, demonstrating the broader applicability of embezzling catalysts in quantum communications.
\\
\section*{Acknowledgments}
We would like to thank Xiao-Min Hu, and Yu Guo for fruitful discussions. This research is supported by the National Research Foundation, Singapore and A*STAR under its Quantum Engineering Programme (NRF2021-QEP2-02-P03), and by A*STAR under its Central Research Funds and C230917003. Yuqi Li, Junjing Xing, Zhaobing Fan, and Haitao Ma are Supported by the Stable Supporting Fund of National Key Laboratory of Underwater Acoustic Technology (KY12400210010). Zhu-Jun Zheng is supported by the Key Lab of Guangzhou for Quantum Precision Measurement (Grant No. 202201000010), the Guangdong Basic and Applied Basic Research Foundation
(Grant No. 2019A1515011703), the Key Research and Development Project of Guangdong Province (Grant No. 2020B0303300001), and the Guangdong Basic and Applied Basic Research Foundation (Grant No. 2020B1515310016).  Dengke Qu, Lei Xiao, and Peng Xue are supported by the National Key R\&D Program of China (Grant No. 2023YFA1406701), the National Natural Science Foundation of China (Grant Nos. 12025401, 92265209, 12104036, and 12305008), and the China Postdoctoral Science Foundation (Grant Nos. BX20230036 and 2023M730198).

\bibliography{Bib}

\end{document}